\documentclass[11pt]{article}
\pdfoutput=1

\usepackage{amssymb, amsmath} 

 \usepackage{amsthm} 
\newtheorem{theorem}{Theorem} 
\newtheorem{lemma}{Lemma} 
\newtheorem{definition}{Definition} 

\usepackage{microtype}
\usepackage{booktabs}
\usepackage{graphicx}
\usepackage{clrscode_if}
\usepackage{float}
\usepackage{algorithm,algorithmic,mathrsfs}

\usepackage{latexsym}

\usepackage{subfigure}

\PassOptionsToPackage{cmyk}{xcolor}

\def\mod{\ensuremath{{\rm \ mod\ }}}
\def\verts{\ensuremath \mathscr{V}}
\def\arcs{\ensuremath \mathscr{A}}
\def\Gr{\ensuremath \mathscr{G}}
\def\iset{\ensuremath \mathcal{I}}

\usepackage{amsbsy}

\newcommand{\matg}[1]{{\ensuremath \boldsymbol{#1}}}
\def\figwidth{0.65\textwidth}

\begin{document}
\sloppy 

\title{Near-Optimal Distributed Scheduling Algorithms for Regular Wireless Sensor Networks}

\author{K. Shashi Prabh\\
School of Engineering\\ Polytechnic Institute of Porto, Portugal}
\date{}

\maketitle

\begin{abstract}
Wireless sensor networks are normally characterized by resource challenged nodes. Since communication costs the most in terms of energy in these networks, minimizing this overhead is important. We consider minimum length node scheduling  in regular multi-hop wireless sensor networks. We present collision-free decentralized scheduling algorithms based on TDMA with spatial reuse that do not use message passing, this saving communication overhead.  We develop the algorithms using graph-based $k$-hop interference model and show that the schedule complexity in regular networks is independent of the number of nodes and varies quadratically with $k$ which is typically a very small number. We follow it by characterizing feasibility regions in the SINR parameter space where the constant complexity continues to hold while simultaneously satisfying the SINR criteria. Using simulation, we evaluate the efficiency of our solution on random network deployments.
\end{abstract}


\section{Introduction}\label{intro}

We consider the problem of deterministic collision-free channel access in regular Wireless Sensor Networks (WSN). WSNs  typically consist of resource challenged nodes where minimizing the communication overhead is important. In pure TDMA, only one node can transmit at a time. However, in wireless networks  deployed over large enough geographical areas,  multiple transmissions can take place in a single time slot due to fading. This scheduling method is referred to as TDMA with spatial reuse, or, STDMA~\cite{NK85_STDMA}. In this paper, we propose a STDMA-based scheduling method without message passing that satisfies the SINR criteria (defined below) at all receivers.

We (a) show that schedule complexity in hexagonal and square-grid networks in graph-based $k$-hop interference model (defined below) is independent of $n$ and is quadratic in $k$, and (b) specify feasibility regions in the SINR parameter space where the constant complexity continues to hold in the SINR model. In the SINR analysis, $k$, which determines the schedule, becomes a free parameter. It has been pointed out that analysis of regular topology networks are useful since, for general networks, the analytical outcome is limited to asymptotics and order computations. On the other hand, analysis of regular networks gives expressions for leading coefficient besides the scaling law and sheds light on the asymptotics of general topology analysis~\cite{MT05_reg_wirel_netw}. The analysis developed in this paper is exact.

In the primary interference model, transmissions between node pairs $(s, t)$ and $(u,v)$ can not take place simultaneously if $\{s,t\} \cap \{u,v\} \neq \phi$. In practice, successful reception depends on the exclusion of a larger set of links from transmitting simultaneously than just those specified by the primary interference model. More accurate graph-based interference models also exclude the nodes that fall within certain proximity of a link from transmitting simultaneously. In the $k$-hop interference model, which is used in this work, the minimum separation between a receiver and simultaneous transmitters is required to be $k$ hops. Hajek and Sasaki showed that the complexity of finding the minimum length schedule is strongly polynomial under primary interference model~\cite{HS88_link_sched_poly}.  However, the minimum length scheduling problem has been shown to be intractable in general graph-based interference models ($k > 1$)~\cite{SV82_complexity_max_matching, SMS06_wless_sched_complexity}. In this work, we present constant complexity algorithms that achieve approximation ratio of 4/3 for hexagonal networks and 5/4 for square-grid networks. Our algorithm for square-grid networks is optimal for odd $k$.

SINR is a more accurate model for successful packet reception than the graph-based models. Let $P_i$ denote the transmit power of node $i$, $G_{ij}$ denote the channel gain between nodes $i$ and $j$, $\eta_r$ denote noise at node $r$ and $\verts$ denote the set of all nodes. In the SINR model, transmissions from $s$ to $r$ are successful if:
\begin{eqnarray}
{\rm SINR}_{sr} &\triangleq& \frac{\rm Signal}{\rm Interference + Noise}  \nonumber\\
&=& \frac{P_s G_{sr}}{ \sum_{i \neq s; i \in \verts}P_i G_{ir} + \eta_r } \ge \beta.\label{eq:sinr}
\end{eqnarray}
A commonly used channel gain model is a polynomially decaying one with the Euclidean distance, $G_{sr} = d(s,r)^{-\gamma}$, where $d(i,j)$ is the distance between nodes $i$ and $j$, and $\gamma$ is the path-loss exponent.\footnote{The value of the path-loss exponent, usually between 2 (free space) and 6, depends on the physical environment~\cite{Rappaprt01_wless_book}. For typical sensor network deployments,  $\gamma$ has been found to be around 3 -- 4.} Thus, whether a given pair of nodes may be able to communicate with each other depends on {\em all\/} other transmitters that are active simultaneously. The notion of connectivity becomes dynamic, and as a result, some of the algorithmic conveniences offered by the graph-based models do not exist anymore~\cite{Peleg_sinr_maps, KLPP11_sinr_wless_topology}. 

In this paper, we develop $k$-hop interference model based near throughput-optimal scheduling algorithms and subsequently derive conditions on the SINR parameters in order to satisfy the SINR criteria at every receiver. Schedule complexity refers to the shortest schedule length that satisfies the traffic demands. The tightest bound on collision-free SINR schedule complexity obtained so far is $I_{in} O(\log^2n)$, where $I_{in} \in \Omega(\log n)$ is a measure of interference and $n$ is the number of nodes~\cite{MWZ06_complexity_sinr}. The quantity $I_{in}$ can be as large as $n$ making this bound trivial. Using our scheduling method gives SINR schedule complexity  independent of $n$.

In this paper, we make the following contributions: we present collision-free distributed scheduling algorithms for node scheduling problems in multi-hop hexagonal and square-grid networks. In the $k$-hop interference model, these algorithms achieve approximation ratios of 3/4 and 5/4 for hexagonal and square-grid networks, respectively. The scheduling is done based on logical address of nodes and  does not use message passing. We show that the schedule complexity depends on $k$ but not on the number of nodes in the network.  We establish feasibility regions where the SINR criteria is also satisfied, thereby establishing constant (in $n$) SINR schedule complexity within these feasible regions. We evaluate our work using simulations.

The rest of this paper is organized as follows: we present model and problem formulation in Section~\ref{sec:prelim}, followed by distributed node scheduling algorithms and schedule complexity analysis in Section~\ref{sec:nodealgo}. We consider SINR scheduling in Section~\ref{sec:sinr}. We present simulation results in Section~\ref{sec:eval}. We present related work in Section~\ref{sec:rel-work} and conclude the paper with Section~\ref{sec:concl}.

\section{Preliminaries and Model}\label{sec:prelim}

We consider two kinds of  regular networks, viz., hexagonal and square-grid.  Some authors use the term ``regular wireless networks'' to refer to a stricter assumption of the regularity of node placement, where the euclidean distance between neighboring nodes is a constant, {\em and\/} of the regularity of connectivity, where each node is connected to a constant (6 and 4 resp.) number of neighbors~\cite{MT05_reg_wirel_netw}. In this paper, we assume regularity with respect to connectivity only~(\figurename~\ref{fig:hexnet}). In the following, we use figures showing the lattice representations (\figurename~\ref{fig:hexlatt}--\ref{fig:sqlatt}) for the convenience of presentation only.

\begin{figure}[!h]
\centering 
\subfigure[Hexagonal networks: physical representation] 
{
\label{fig:hexnet}
\includegraphics[width=1.4in]{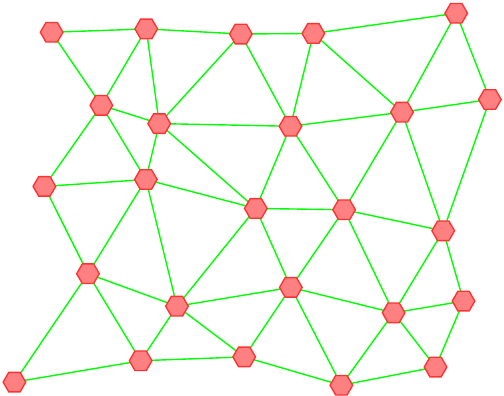}
}
  \hspace{0.05\columnwidth}
\subfigure[Hexagonal networks: lattice representation] 
{
\label{fig:hexlatt}
\includegraphics[width=1.25in]{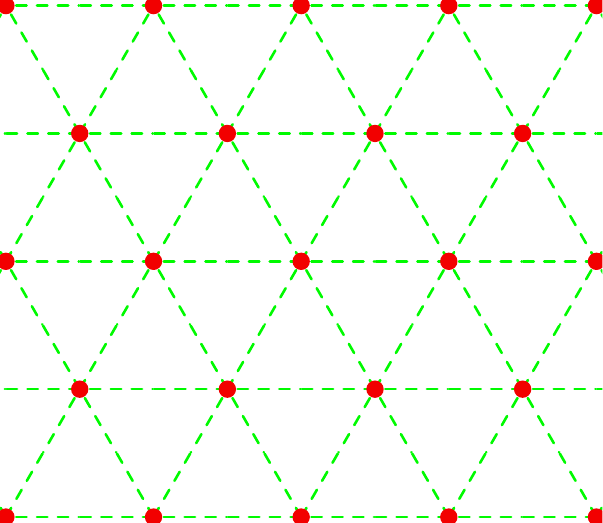}
}
  \hspace{0.05\columnwidth}
\subfigure[Square-grid networks: lattice representation] 
{
  \label{fig:sqlatt}
  \includegraphics[width=1.1in]{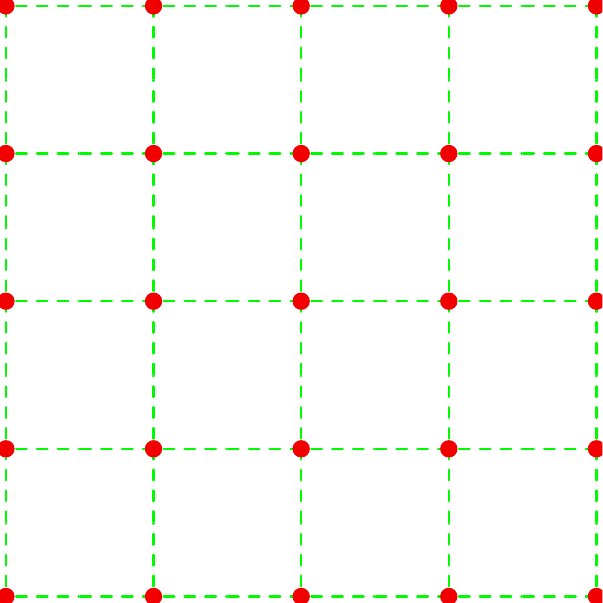}
}
\caption{Physical and logical representations}
\label{fig:lattices} 
\end{figure}

The coordinates of  nodes in  lattice representations are convenient for addressing. In the case of hexagonal lattice, a few coordinate systems have been used in the literature. In this paper, we choose two diagonals oriented at an angle of $2\pi/3$ as the  $X$ and $Y$ axes~(\figurename~\ref{fig:hexaxis}).  The corresponding addressing is shown in~\figurename~\ref{fig:hexcoord}. We refer to the six regions enclosed by the three diagonals as {\em hextants\/}. In square-grid networks, addresses of the four neighbors of $(i,j)$ are $(i \pm 1, j \pm 1)$.

\begin{figure}[!h]
\centering 
\subfigure[] 
{
\label{fig:hexaxis}
\includegraphics[width=1.6in]{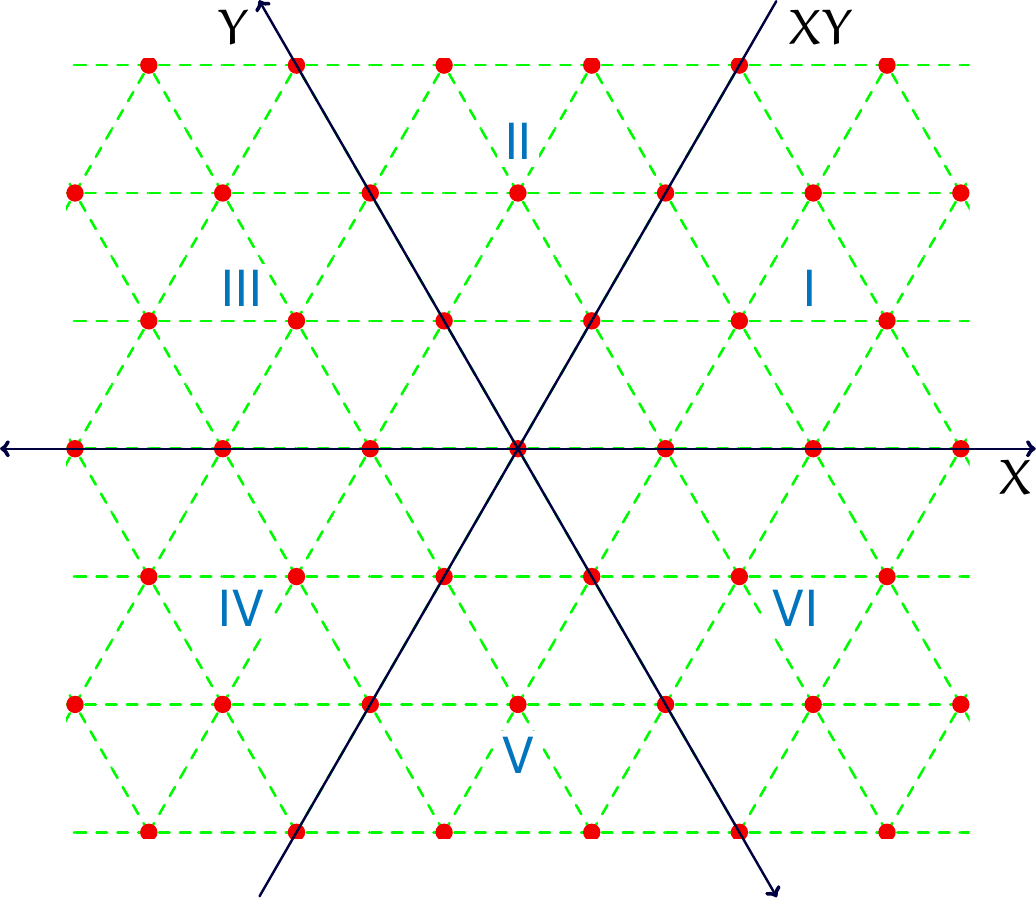}
}
\hspace{0.05\columnwidth}
\subfigure[] 
{
  \label{fig:hexcoord}
  \includegraphics[width=2in]{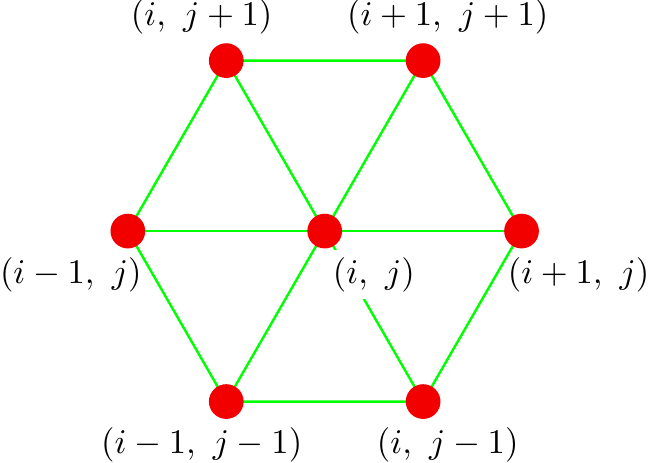}
}
\caption{Hexagonal coordinate system. }
\label{fig:hexcoordall} 
\end{figure}

\begin{lemma}\label{lem:hexdist}
The graph distance (i.e., the minimum number of edges) between points $(x_1, y_1)$ and $(x_2, y_2)$ in hexagonal lattice is given by $MAX\{|dx|,|dy|,|dx-dy|\}$, where $dx = x_2 - x_1$ and $dy = y_2 - y_1$.
\end{lemma}
A variation of Lemma~\ref{lem:hexdist} appears in~\cite{DS05_3d_hex}. The difference in the expression for distance in this paper and in~\cite{DS05_3d_hex} is due to different coordinate system choices. A proof of Lemma~\ref{lem:hexdist} appears in the appendix.

\begin{lemma}\label{lem:sqdist}
The graph distance ($L_1$ norm) between points, $(x_1, y_1)$ and $(x_2, y_2)$ in square lattice is given by $|dx|+|dy|$, where $dx = x_1 - x_2$ and $dy = y_1 - y_2$.
\end{lemma}

\subsection{Problem Formulation}\label{sec:prob}

The \textbf{\em $k$-hop interference model\/} is defined as follows: a transmission is successful if the shortest hop distance between the receiver and all concurrent transmitters is greater than or equal to $k$ (note that for $k=1$, primary interference constraints need to be taken into account explicitly).  

Let  $\Gr=(\verts,\arcs)$ be a directed graph, where $\verts$  and $\arcs$ represent the set of nodes and arcs respectively. A subset $a \subset \arcs$ is called a {\em matching\/} if all the links in $a$ can be active simultaneously. A node schedule consists of a discrete sequence of set of nodes $\mathcal{N}_1, \ldots, \mathcal{N}_m$ where every combination of one arbitrary link from each node in  $\mathcal{N}_i$ forms a matching. Let $\iset$ be a finite index set. The schedule is defined as $\mathcal{S} = (\mathcal{N}_i,\lambda^i | i \in \iset)$ where $\lambda^i$ is the duration of time slot $i$.  
Let
\begin{eqnarray} 
\chi_{v}^i   & = & 1  ~~\text{if}~~ v \in \mathcal{N}_i\nonumber\\
             & = & 0,  ~~\text{otherwise.} 
\end{eqnarray}
%
%
The Minimum-Length Node Scheduling Problem {\bf [NSP]} is defined as:
\begin{eqnarray} 
\text{Minimize}   &~~ \tau = \sum_i  \lambda^i\nonumber\\
\text{subj. to}   &~~ \sum_i   \chi_{v}^i \lambda^i  \ge 1,~~  \forall v \in \verts.\label{nsp:cond}
\end{eqnarray}

\subsection{Basis Lattice Section}

Intuitively, a basis lattice section is a geometric region, replication of which covers the entire lattice without overlapping.
\begin{definition}[Basis Lattice Section]
A  section $\psi$ of a lattice is defined to be basis lattice section if linear translations of the replicas of $\psi$ tessellate the lattice, satisfying the following two conditions:
\begin{enumerate}
\item Coverage: Every lattice point lies within some replica of $\psi$.
\item Uniqueness: No lattice point lies within more than one replica of $\psi$.
\end{enumerate}
\end{definition}
For example, a shaded rhombus in \figurename~\ref{fig:hex:prim:2hop} and a shaded rectangle in \figurename~\ref{fig:sq:prim:3hop} are basis lattice sections of the hexagonal and square lattices, respectively. Similar ideas have been used in the literature, particularly in cellular networks (see e.g.,~\cite{FDTT07_wless_capacity_percolation}). An important difference, however, is that in this work  {\em we are not tessellating the physical space.}

\subsubsection*{Notations and conventions}

\begin{itemize}
\item Let a basis lattice section $\psi$ contain a lattice point $p$. $\matg{N_l}(p)$ denotes the set of lattice points in $\psi$ at graph distance $l$ from $p$ :   $\matg{N_l}(p) = \{ q \in \psi ~\vert~ d(p,q) = l\}$, where $d(p,q)$ denotes the graph distance between $p$ and $q$.

\item Let $p \in \psi$ be the lattice point with the smallest coordinate in every dimension. Then we call $p$ as the origin of $\psi$, and use $p$ to refer to $\psi$. For example, we refer to a quadrilateral basis lattice section whose corner points are at $(0,0), (i,0), (i, i), (0,i);~ i > 0$ as ``the basis section at the origin.''

\item Let $p \in \psi$. Then $\matg{S}(x,y)$ denotes the set of lattice points $\{q\}$ such that the relative coordinates of $q$ (w.r.t. the basis lattice section containing $q$) are the same as that of $p$. For example,  the elements of $\matg{S}(0,0)$ are shown with circled dots in \figurename~\ref{fig:hex:arb:sched}.
\end{itemize}

\begin{figure}[h]
\centering
\subfigure[Rhombus basis section.]
{
\includegraphics[height=1.8in]{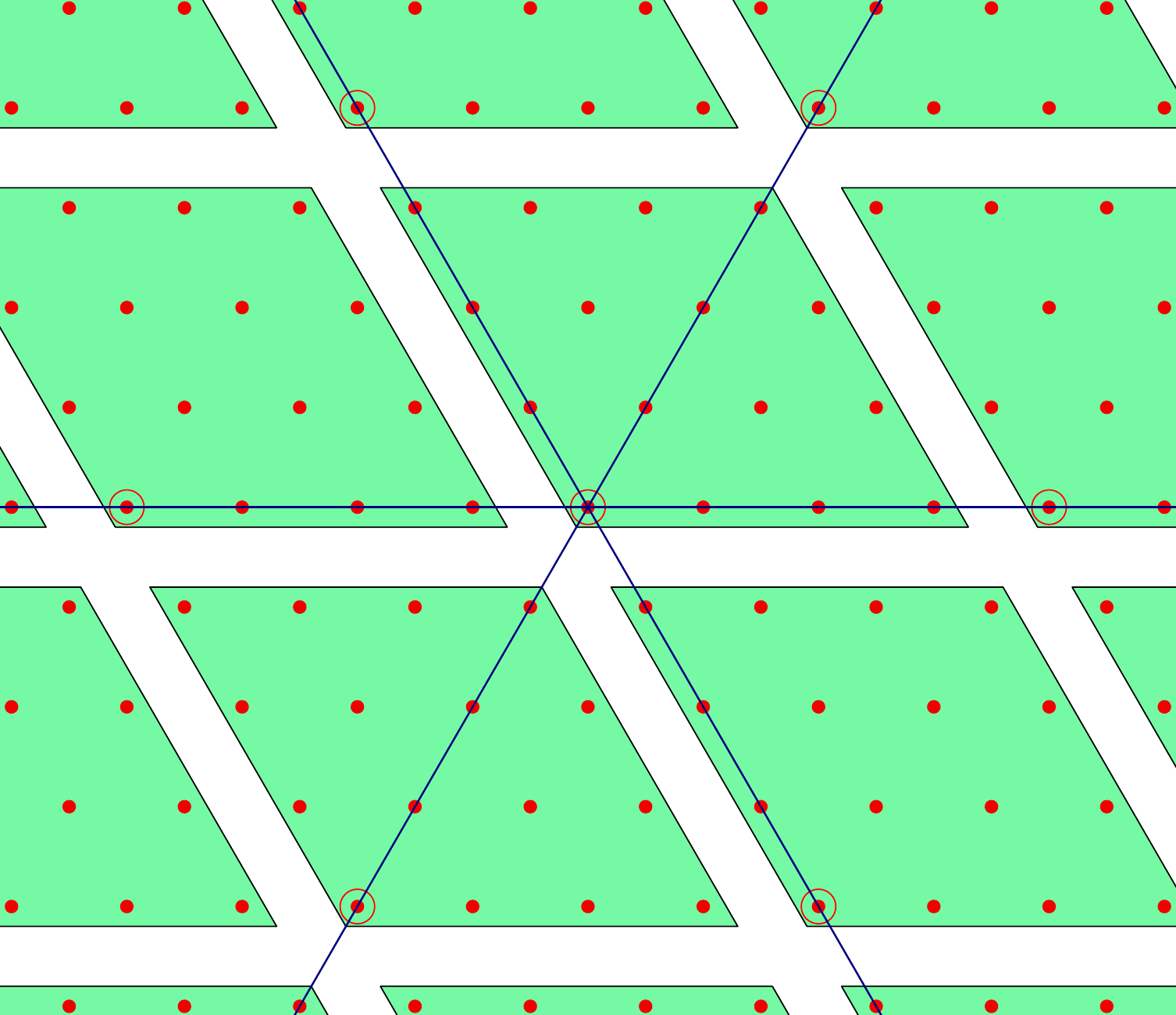}
\label{fig:hex:prim:2hop}
}
 \hspace{0.05\columnwidth}
\subfigure[Rhomboid basis sections.]
{
  \label{fig:sq:prim:3hop}
  \includegraphics[height=1.7in]{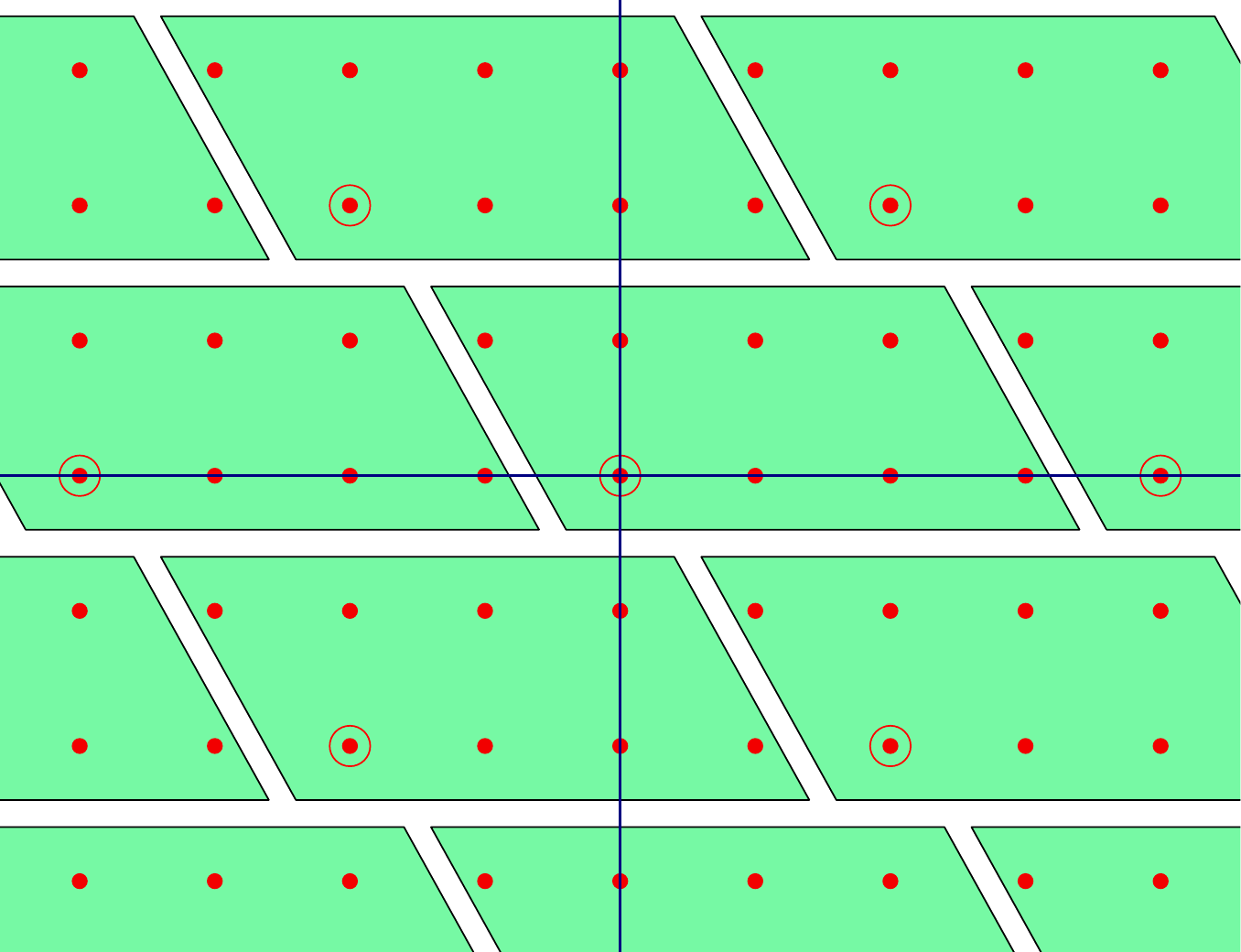}
}
\caption{Basis lattice sections of hexagonal and square lattices. The circled points represent one set of concurrent transmitters for  3-hop interference model.}
\label{fig:hex:arb:sched} 
\end{figure}

\section{Scheduling Algorithms}\label{sec:nodealgo}
In this section we present node scheduling algorithms and establish their approximation ratios. Our approach consists of determining a maximal set of nodes that can be scheduled simultaneously along-with a suitable basis lattice section where every basis lattice section contains one scheduled node. The number of lattice points contained in the basis section gives the schedule length. We use the clique number of interference graphs to analyze the optimality of our algorithms.  

\subsection{Hexagonal Networks}
Consider a rhombus of $i$ edges drawn parallel to the two axes whose corner points are $(0,0), (i,0), (i, i)$ and  $(0,i)$ (\figurename~\ref{fig:hex:prim:2hop}).  Clearly, this rhombus satisfies the conditions for a basis lattice section.

\begin{definition}[Rhombus Basis Section]
In the following, we refer to the basis lattice section of hexagonal lattice whose corner points are of the form $(x,y), (x+i,y), (x+i, y+i)$ and $(x,y+i)$ as rhombus basis section.
\end{definition}

{\bf Scheduling algorithm.~} Algorithm~\ref{hx-nack-algo1} solves NSP for hexagonal networks in $(k+1)^2$ slots. Starting with 
$\{(0,0)\}\cup\matg{S}(0,0)$, it progressively schedules nodes in $\{(k,k)\}\cup\matg{S}(k,k)$, thus covering all nodes in the network. 
 An illustration of the slot sequence for the 3-hop interference case is shown in \figurename~\ref{fig:hex:arb:sched1}. Observe that the scheduling decision in a given slot is solely based on the (lattice) address of a node and hence does not involve any message passing.

\begin{algorithm}
  \caption{Node scheduling algorithm for hexagonal networks}
  \label{hx-nack-algo1}
  \begin{algorithmic}[1]
    \REQUIRE $(x,y)$ \COMMENT{Node address}
    \REQUIRE $k$ \COMMENT{Interference model} 
    \STATE $u \leftarrow x \mod (k+1)$
    \STATE $v \leftarrow y \mod (k+1)$
    \ENSURE $t(x,y) \leftarrow u + (k+1) v$ \COMMENT {Tx slot}  
  \end{algorithmic}
\end{algorithm}
\subsubsection* {Correctness proof.}
\begin{lemma}\label{lem:hexbasicdist}
In a rhombus basis section $\psi$ of length $i$ located at $(x,y)$, the maximum of the minimum graph distance between the lattice points in $\matg{N_1}(x,y)$ and any other lattice point in $\psi$ is $i-1$.
\end{lemma}

\begin{proof}
To simplify the proof, we translate the origin to $(x,y)$. Thus it is sufficient to show that the lemma is true for the basis section at the origin.
$\matg{N_1}(0,0) = \{ (0,1), (1,1), (1,0)\}$. From Lemma~\ref{lem:hexdist}, the lattice points in the rhombus basis section at distance $l$ from $(0,0)$ are $\matg{N_l}(0,0) = \{(l-i, l), i = 0, \ldots, l\} \cup \{ (l,l-i), i = 1, \ldots, l \}$. Applying Lemma~\ref{lem:hexdist} once again, the minimum distance between 
 $\matg{N_1}(0,0)$ and $\matg{N_l}(0,0)$ is $l-1$. Since in a rhombus basis section of length $i$, the maximum value of $l = i$, the lemma follows.
\end{proof}



Consider a layout of rhombus basis sections of side-length $k$ where the origins of the sections lie at $(n(k+1), m(k+1))$ where $n, m$ are integers~(\figurename~\ref{fig:hex:prim:2hop}). WLOG, consider the rhombus basis section at the origin and let the node at the origin be scheduled for transmission.  From Lemma~\ref{lem:hexbasicdist},  the graph distance between $\matg{N_1}(0,0)$ and any other point in the section is at-most $k-1$. Since the graph distance between the nearest neighbors in $\matg{S}(0,0)$ is $k+1$, the number of edges between all six neighbors of $(0,0)$ and the nodes in $\matg{S}(0,0)$ is at-least $k$. Therefore, all nodes in $\matg{S}(0,0)$ can  be scheduled simultaneously.

 Algorithm~\ref{hx-nack-algo1} schedules nodes in $\matg{S}(0,0)$ in the first time slot, $t = 0$, followed by those in $\matg{S}(1,0)$ in the second time slot and so on till $\matg{S}(k,0)$ in $t = k$.  It follows the same pattern for all rows in a rhombus basis section: starting with nodes in $\matg{S}(0,i)$, that are scheduled in $t = ik+k$, it schedules the last node of this row in $t = (i+1)k+k$.  Since there are $(k+1)^2$ sets of the form $\matg{S}(i,j)$ with respect to the rhombus basis section, using the coverage property of the basis sections, we conclude that $(k+1)^2$ time slots are sufficient. Thus we can state the following theorem:

\begin{theorem}\label{thm:hex-interf} 
The NSP schedule complexity for hexagonal networks is bounded from above by $(k+1)^2$.
\end{theorem}

\begin{figure}[h]
\centering
\subfigure[Hexagonal networks]
{
  \label{fig:hex:arb:sched1}
  \includegraphics[width=2.2in]{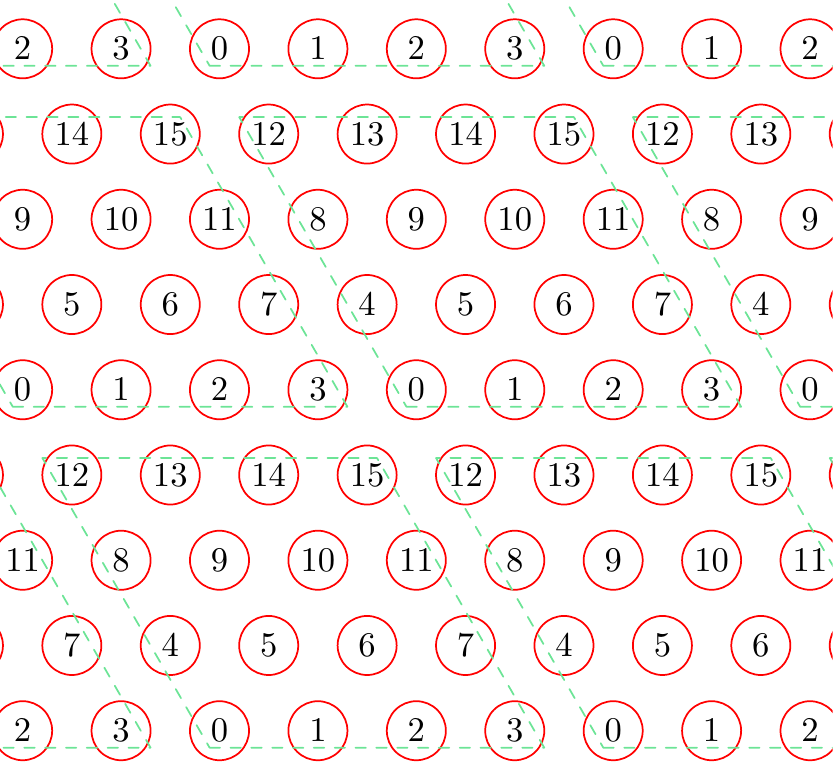}
}
%
\subfigure[Square-grid networks]
{
  \label{fig:sq:arb:sched1}
  \includegraphics[width=2.25in]{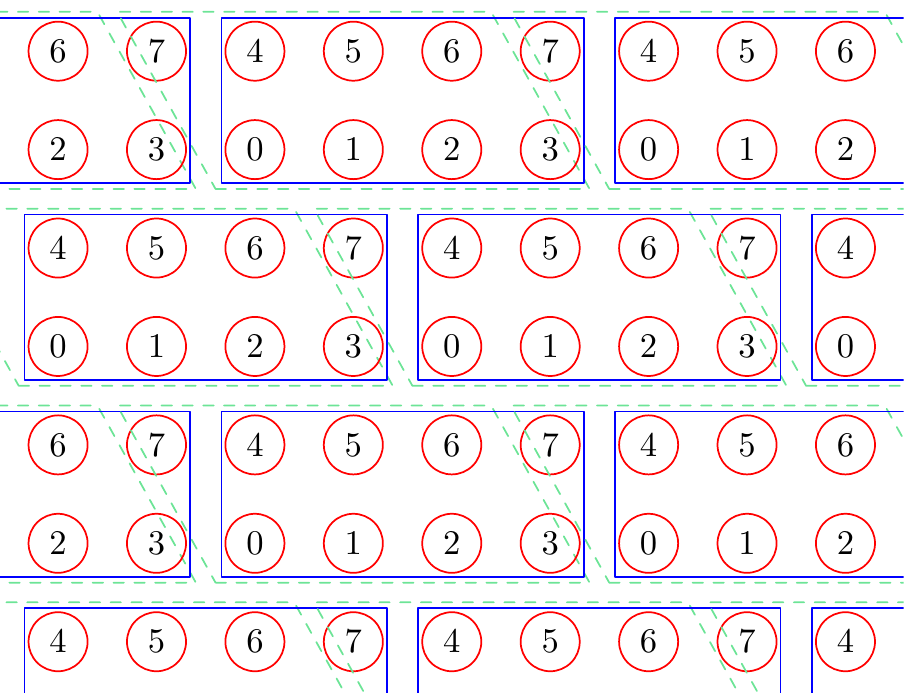}
}
\caption{Schedules for 3-hop interference model.}
\label{fig:sq:arb:all} 
\end{figure}

\subsection{Square-Grid Networks}

The intuition behind the treatment for square-grid networks is as follows. Consider the $3$-hop interference model. Let the origin of the square lattice be the location of one of the transmitters~(\figurename~\ref{fig:sq:prim:3hop}). Then the nodes at $(\pm 1, 0)$ and $(0, \pm 1)$ are the receivers. Clearly, nodes at $(\pm 4, 0)$ can be  scheduled simultaneously since their minimum distance from the neighbors of $(0,0)$ is exactly 3. However, along the $Y$-axis, we can shrink the rhomboid basis section by placing transmitters not at $(0, \pm 4)$, but at $(\pm 2, \pm 2)$, thereby decreasing the schedule length.
\begin{definition}[Rhomboid Basis Section]
Rhomboid basis section of size $(i,j)$ is a basis lattice section of square lattice whose corner points are of the form $(x,y), (x+i,y), (x+i-j, y+j), (x-j,y+j)$.
\end{definition}

{\bf Scheduling algorithm.~} Algorithm~\ref{sq-nack-algo1} solves NSP for square-grid networks using $(k+1) \lceil (k+1)/2 \rceil$ slots. Unlike the rhombus basis sections for hexagonal networks, the rhomboid basis sections are not aligned parallel to the (rectangular) axes. Hence, to further simplify the algorithm, we rearrange the slot sequences to align with the axes: pictorially, we are simply cutting the  upper left triangle of the rhomboids and placing them to the right side of the resulting trapezoids to complete a rectangle. In both cases, the number of lattice points remain the same. In (\figurename~\ref{fig:sq:arb:sched1}), the dashed lines mark the rhomboid basic sections and the solid lines mark the rectangles showing the sequence in which Algorithm~\ref{sq-nack-algo1} assigns time-slots.

Similar to the hexagonal case, the algorithm selects the nodes for transmission, $\matg{S}(x,y)$, systematically for all lattice points $(x,y)$ in every basis section. The rectangular basis sections are shifted along $X$-axis by $\lceil k/2 \rceil$ alternately (\figurename~\ref{fig:sq:arb:sched1}). The presence of shifts is indicated by the Boolean variable $b$ (Line 1). 

\begin{algorithm}
  \caption{Node scheduling algorithm for square-grid networks}
  \label{sq-nack-algo1}
  \begin{algorithmic}[1]
    \REQUIRE $(x,y)$  \COMMENT{Node address}
    \REQUIRE $k$ \COMMENT{Interference model}
    \STATE $b \leftarrow \left\lfloor \left( y/ \left\lceil \frac{(k+1)}{2} \right\rceil \right)\right\rfloor \mod 2$ 
    \STATE $u \leftarrow \left(x + \left(b * \lceil (k+1)/2 \rceil\right)\right) \mod (k+1)$
    \STATE $v \leftarrow y \mod \left\lceil (k+1)/2 \right\rceil$
    \ENSURE $t(x,y) \leftarrow u + (k+1) v$  \COMMENT{Tx slot}
  \end{algorithmic}
\end{algorithm}

\subsubsection* {Correctness proof.}

\begin{lemma}\label{lem:sqbasicdist}
In a rhomboid basis section of size $(i,\lceil i/2 \rceil)$ located at $(x,y)$, the maximum of the minimum graph distance between lattice points in $\matg{N_1}(x,y)$ and any other lattice point in the rhombus basis section is $i-1$.
\end{lemma}

\begin{proof}
Similar to the proof of Lemma~\ref{lem:hexbasicdist}. WLOG we translate the origin to $(x,y)$. Thus it is sufficient to show that the lemma is true for the rhomboid basis section at the origin.
$\matg{N_1}(0,0) = \{ (0,1), (1,0)\}$. From Lemma~\ref{lem:sqdist}, the coordinates of the lattice points in the rhomboid basis section satisfy $|x|+|y| \le i$. Applying Lemma~\ref{lem:sqdist} once again, the maximum of the minimum edge count between $\matg{N_1}(0,0)$ and a lattice point in the basis section is $i-1$.
\end{proof}

Consider the node at the origin. Clearly, it can be scheduled along-with the set of nodes at $(\pm n(k+1), 0)$ since they are at $\ge k$ edges from the neighbors of $(0,0)$. Let $P(u,v), v \neq 0$ be the nearest node from the origin that can be simultaneously scheduled.  WLOG, let $P$ be in the first quadrant where  $v > 0$. Then, $P$ must be $\ge k$ edges from $(0,1)$, $(1,0)$, $(k,0)$ and $(k+1,1)$, where the latter two are the neighbors of $(k+1,0)$. From Lemma~\ref{lem:sqdist}, we get the following two unique conditions:
\begin{eqnarray}
u + v -1 &=& k \nonumber\\
k - u + v &=& k  \nonumber\\
\Rightarrow v &=& \left\lceil (k+1)/2 \right\rceil ~=~ u.
\end{eqnarray}
Thus, we have shown that the origins of rhomboid basis sections defined above can be scheduled simultaneously, and more generally, the lattice coordinates of the six nearest neighbors of a node at $(x,y)$ that can be scheduled simultaneously are $(x \pm (k+1), y), (x \pm u, y \pm v)$~(\figurename~\ref{fig:sq:prim:3hop}). The rest of the correctness proof can be completed in a similar manner as done above for the hexagonal networks.

\begin{theorem}\label{thm:sq-interf}
The NSP schedule  complexity for square-grid networks is bounded from above by $(k+1) \lceil (k+1)/2 \rceil$.
\end{theorem}

\subsection{Optimality Analysis}
 Interference graph of $\Gr=(\verts,\arcs)$ is defined to be $\Gr'=(\verts,\arcs')$ where $(u,v) \in \arcs'$ if simultaneous transmissions of the nodes $u, v \in \verts$ interfere. Clearly, the schedule complexity is bounded from below by the clique number (number of nodes in a maximum clique) of the interference graph. Let $\omega(\Gr'_H,k)$ and $\omega(\Gr'_S,k)$ denote the clique numbers of the interference graphs for the $k$-hop interference model of hexagonal and square-grid networks, respectively. We use subscripts $e$ and $o$ to denote the cliques for even and odd $k$, respectively. \figurename~\ref{fig:clik} shows some maximum cliques where we have used the symbols for clique numbers to label corresponding figures.

\begin{figure}[h]
\centering
\subfigure[$\omega(\Gr'_H,2)$]
{
  \label{fig:clik:h2}
  \includegraphics[scale=0.4]{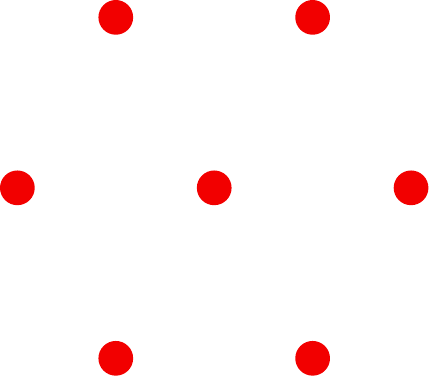}
}
  \hspace{0.05\columnwidth}%
\subfigure[$\omega(\Gr'_H,3)$]
{
  \label{fig:clik:h3}
  \includegraphics[scale=0.4]{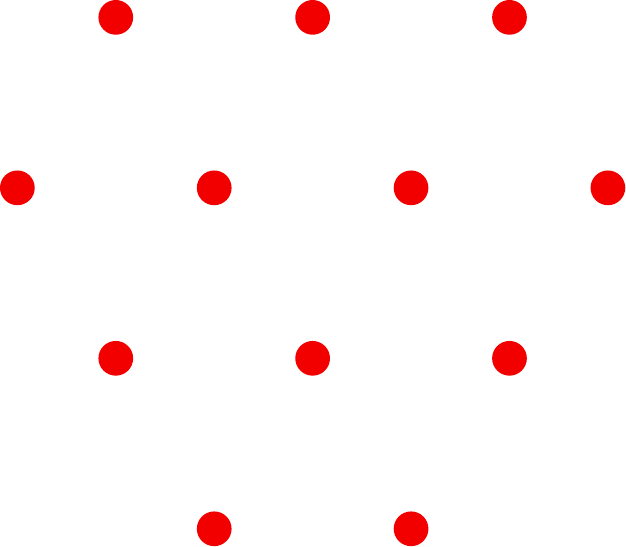}
}
\hspace{0.05\columnwidth}
\subfigure[$\omega(\Gr'_s,2)$]
{
  \label{fig:clik:s2}
  \includegraphics[scale=0.4]{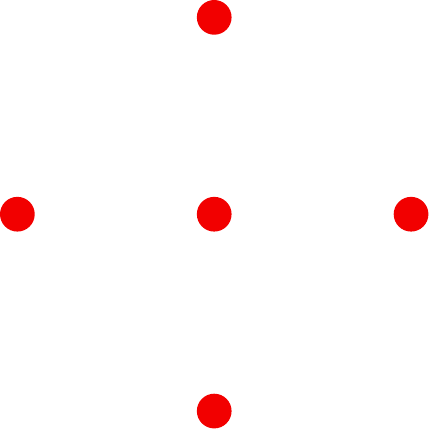}
}
\hspace{0.05\columnwidth}
\subfigure[$\omega(\Gr'_s,3)$]
{
  \label{fig:clik:s3}
  \includegraphics[scale=0.4]{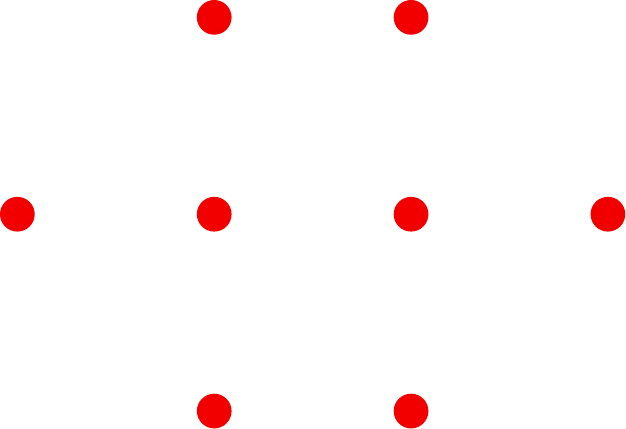}
}
\caption{Maximum cliques}
\label{fig:clik} 
\end{figure}

For a hexagonal interference graph and even $k$,  a complete hexagonal network of diameter $k$ ($k/2$ concentric hexagonal rings where the $i^{\rm{th}}$ ring has $6i$ nodes plus a node at the center) forms a maximum clique (\figurename~\ref{fig:clik}). For odd $k$, a complete hexagonal network of diameter $k-1$ and their neighbors on any one side of a diagonal (excluding the nodes on the diagonals) forms a maximum clique. Therefore, $\omega(\Gr'_H,k)$ for even $k$ is:

\begin{eqnarray}
\omega_e(\Gr'_H,k) &=& 1 + \sum_{i=1}^{k/2} 6i\nonumber\\
                 &=& 3 (k/2)^2 + 3 (k/2) + 1,
\end{eqnarray}

and that for odd $k$ is:
\begin{eqnarray}
\omega_o(\Gr'_H,k) &=& 1 + \sum_{i=1}^{(k-1)/2} 6i + 3 (k+1)/2 -1\nonumber\\
                 &=& 3 (k/2)^2 + 3 (k/2) + 2.
\end{eqnarray}

Interestingly, maximum clique in square-grid interference graphs is formed by node arrangements analogous to the hexagonal interference graphs. For even $k$, $k/2$ complete concentric squares oriented at an angle $\pi/4$ and for odd $k$, $(k-1)/2$ complete concentric squares oriented at an angle $\pi/4$ plus neighboring nodes on any one side of an axis form a maximum clique (\figurename~\ref{fig:clik}). Thus, it can be shown that:
\begin{eqnarray}
\omega_e(\Gr'_S,k) &=&  k^2/2  +  k + 1,\\
\omega_o(\Gr'_S,k) &=& (k+1)^2/2.
\end{eqnarray}

As we presented earlier, the exact schedule lengths of hexagonal and square-grid node scheduling algorithms are  $(k+1)^2$ and  $(k+1) \lceil (k+1)/2 \rceil$, respectively. Since  $(k+1)^2  = (4/3) \omega(\Gr'_H,k) - 1/3$ for even $k$, and  $(k+1)^2  = (4/3) \omega(\Gr'_H,k) - 5/3$ for odd $k$, the approximation ratio of the hexagonal scheduling algorithm is $4/3$. 

For square-grid networks and even $k$, $(k+1) (k+2)/2  < \rho * \omega(\Gr'_S,k)$, where $\rho$ is a non-increasing function of $k$, and hence attains its maximum 5/4 at $k=2$ (note that $k > 0$). For odd $k$, the approximation ratio is 1. Thus, the approximation ratio of the square-grid scheduling algorithm is $5/4$ and the algorithm is optimal for odd $k$'s.

\section{Scheduling with SINR Constraints}\label{sec:sinr}
We derive conditions such that the SINR criteria is satisfied at all receivers when the scheduling algorithms presented in the previous section are used. 
%
%
%
%
 We make the non-limiting assumption of the path-loss exponent, $\gamma > 2$. We also assume uniform noise level. Let $d$  and $D$ be the minimum and maximum euclidean distances between a pair of neighboring nodes in a given deployment. Since the sender-receiver separation is at-most $D$, the lower bound on the received signal strength is given by: 
\begin{equation}
{\rm Signal}_{\rm Min} = \frac{P} {D^{\gamma}}.\label{eq:sigmin}
\end{equation}
Observe that upper bound on interference is attained at a node at the center of a deployment in a {\em physical\/} node placement configuration where {\em the nodes are located on regular lattice points  having the inter-node separation equal to $d$}.


\begin{figure}[h]
\centering 
\subfigure[Hexagonal networks] 
{
\label{fig:hexconc}
\includegraphics[width=2in]{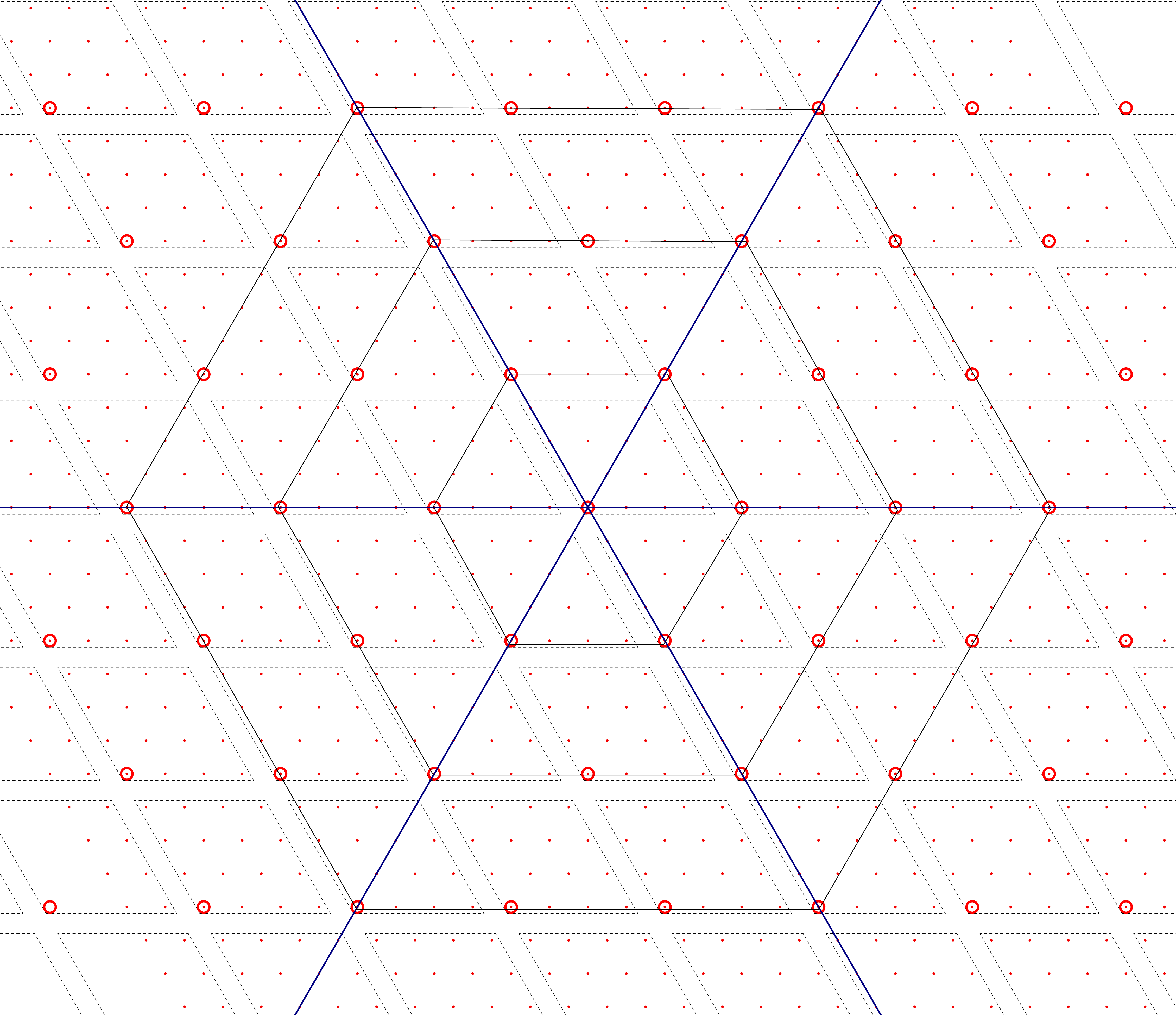}
}
\subfigure[Square-grid networks] 
{
  \label{fig:sqconc}
  \includegraphics[width=2in]{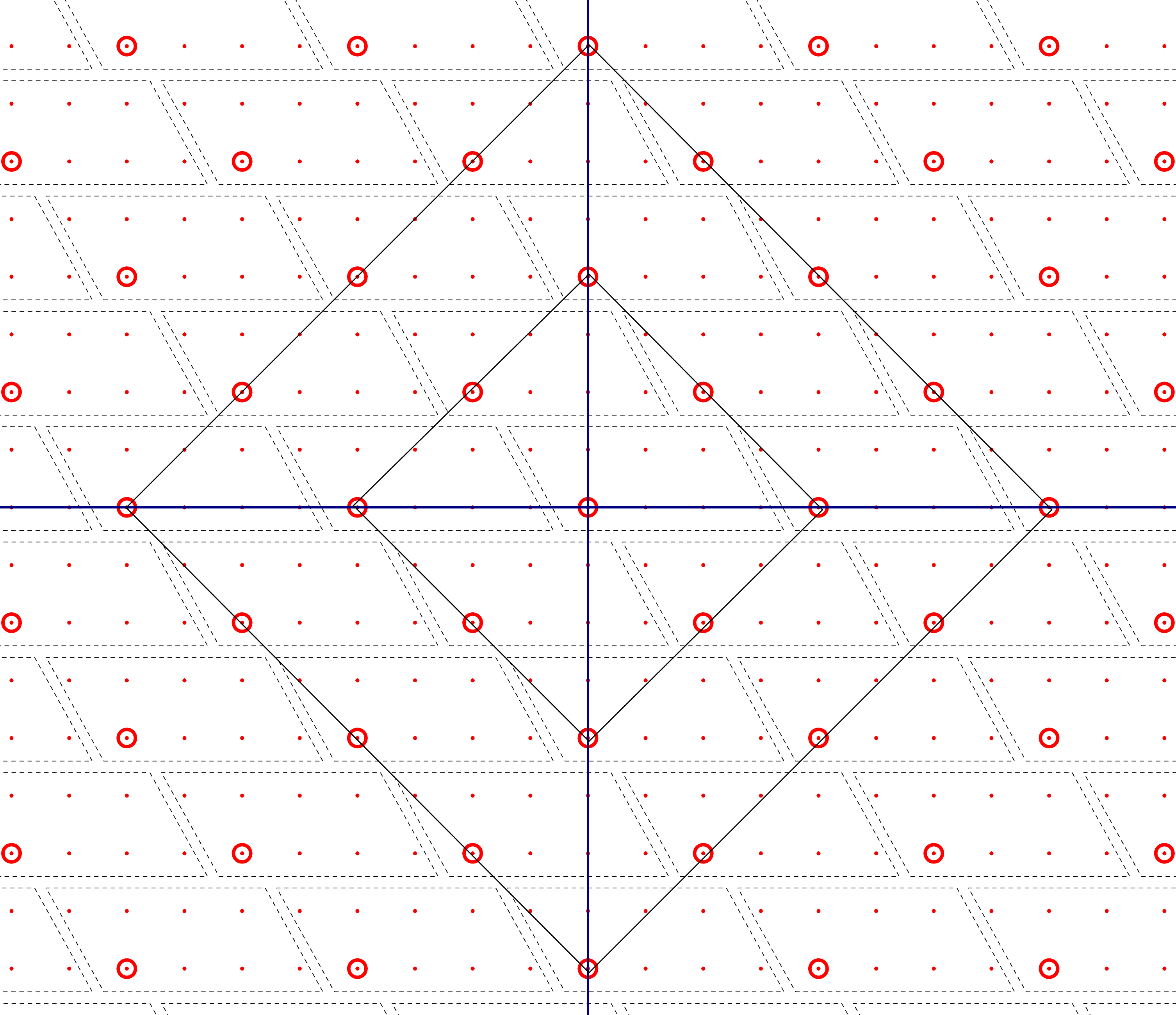}
}
\caption{ A maximal simultaneously scheduled node set (shown with larger dots.)}
\label{fig:maxset} 
\end{figure}

\subsection{Hexagonal Networks} 
Simultaneously scheduled nodes fall on concentric hexagons in the lattice representation, as shown in Figure~\ref{fig:hexconc}. From Theorem~\ref{thm:hex-interf}, the minimum distance between simultaneously scheduled transmitters is $k+1$ edges. Let $l$ denote the lower bound on the corresponding euclidean distance. Thus, $l = (k+1)d$. To get the upper bound on interference at the neighbors of the transmitter at the origin, we use the symmetry of the network. We derive the contribution from one hextant~(Fig.~\ref{fig:hexcoord}) and multiply it by 6. WLOG, consider the $n^{\rm th}$ concentric hexagon in the first hextant. The lattice coordinates of the transmitters are $(n(k+1),0), (n(k+1), k+1), \ldots, (n(k+1), (n-1)(k+1)) $. Let $s_{i,n}$ denote the euclidean distance between the origin and the $i^{\rm th}$ transmitter. Then,
\begin{eqnarray}
s_{i,n}^2  & \ge & n^2 l^2 + (i-1)^2 l^2 - 2 (i-1) n l^2 \cos(\pi/3)\nonumber\\
 {\rm or,~}   s_{i,n}   &\ge& l (n^2 + (i-1)^2 - (i-1) n)^{1/2}\nonumber
\end{eqnarray}

The euclidean distance between any of these $n$ transmitters and any neighbor of the node at the origin is lower bounded by  $s_{i,n} - d$. Therefore, 
 an upper bound on interference is:
\begin{eqnarray}
{\rm Interf.}  &\le& 6 \sum_{n=1}^{\infty} \sum_{i=1}^{n} \frac{P}{{(s_{i,n}-d)}^{\gamma}}.\label{eq:interf:h1}
\end{eqnarray}
%
Expanding and simplifying (\ref{eq:interf:h1}), we get
\begin{eqnarray}
{\rm Interf.}  & < & \frac{6 P}{l^\gamma} \left(\frac{2}{\sqrt{3}}\right)^\gamma \sum_{n=1}^{\infty}  \frac{1}{n^{\gamma-1}}.\label{eq:interf:h2}\\
 &<& \frac{6 P}{{l}^{\gamma}} \left(\frac{2}{\sqrt{3}}\right)^{\gamma} \frac{\gamma -1}{\gamma -2},\label{eq:inthex0}\\
&=& \frac{6 P}{{((k+1)d)}^{\gamma}} \left(\frac{2}{\sqrt{3}}\right)^{\gamma} \frac{\gamma -1}{\gamma -2},\label{eq:inthex}
\end{eqnarray}
 where (\ref{eq:inthex0}) follows from a Riemann-Zeta function bound.

Thus, from (\ref{eq:sigmin}) and (\ref{eq:inthex}), SINR is above threshold, $\beta$, if:
\begin{equation}
P > \frac{\beta \eta D^{\gamma}} {1- 6 \beta {\left(\frac{2 D}{\sqrt{3}(k+1)d}\right)}^{\gamma} \frac{\gamma-1}{\gamma-2}}. \label{eq:hexpossible}
\end{equation}

Since $P > 0$, the denominator must be greater than 0. Therefore:
\begin{equation}
\frac{D}{d} <  \frac{\sqrt{3}(k+1)}{2} {\left(\frac{\gamma-2}{6 \beta (\gamma-1)}\right)}^{1/\gamma}. \label{eq:hexdd}
\end{equation}

Finally, since $D \ge d$, 
\begin{equation}
\beta \le  \left(\frac{\sqrt{3}(k+1)}{2}\right)^\gamma {\left(\frac{\gamma-2}{6 (\gamma-1)}\right)}. \label{eq:hexbta}
\end{equation}

Plots of the feasibility regions given by (\ref{eq:hexdd})-(\ref{eq:hexbta}) are shown in \figurename~\ref{fig:feasr}.

\begin{figure}[!h]
\centering 
\subfigure[Inequality (\ref{eq:hexdd}), $\beta = 1$] 
{
\label{fig:hexDd}
\includegraphics[width=2.25in]{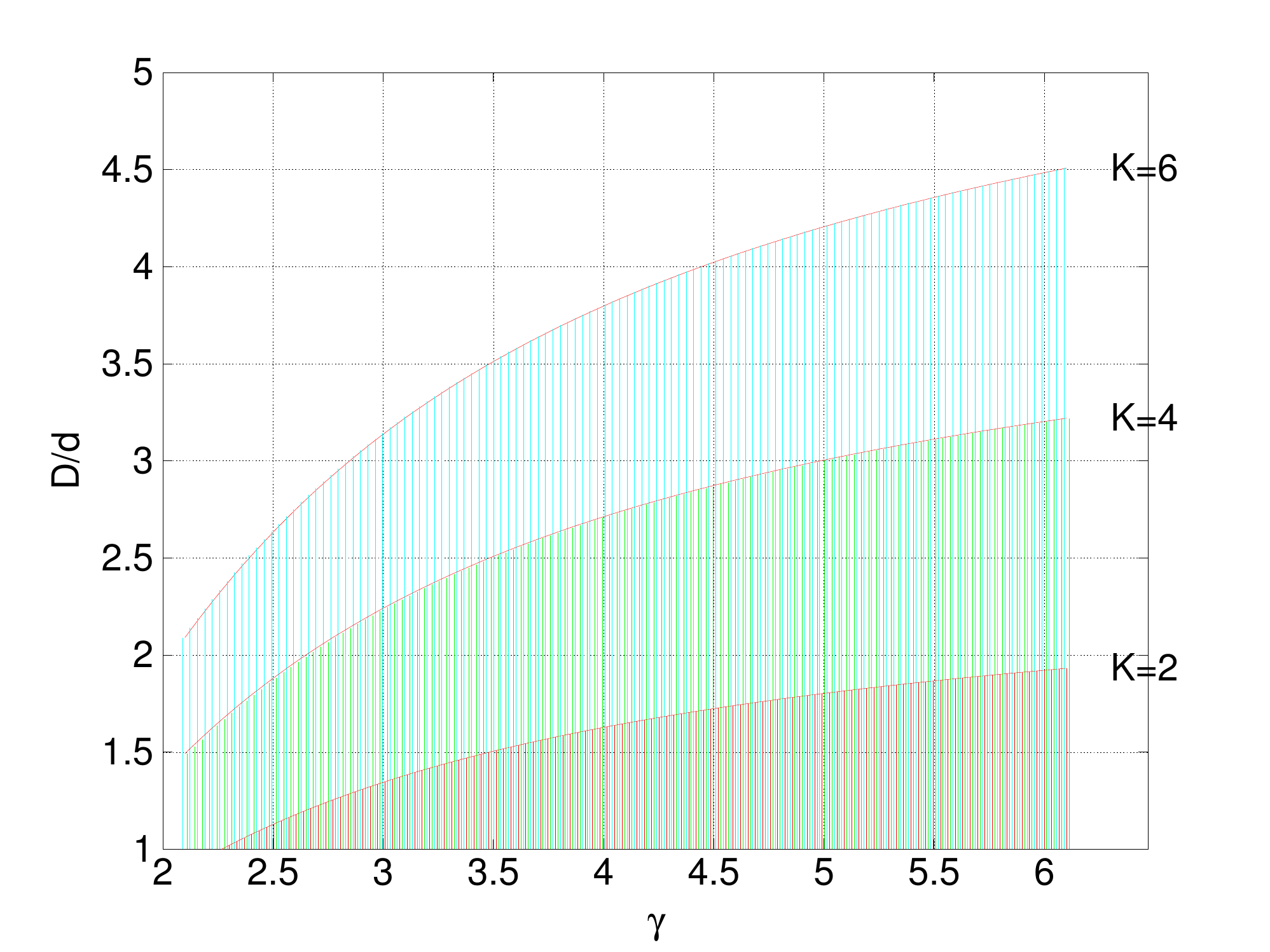}
}
\subfigure[Inequality (\ref{eq:hexbta})] 
{
  \label{fig:hexbeta}
  \includegraphics[width=2.25in]{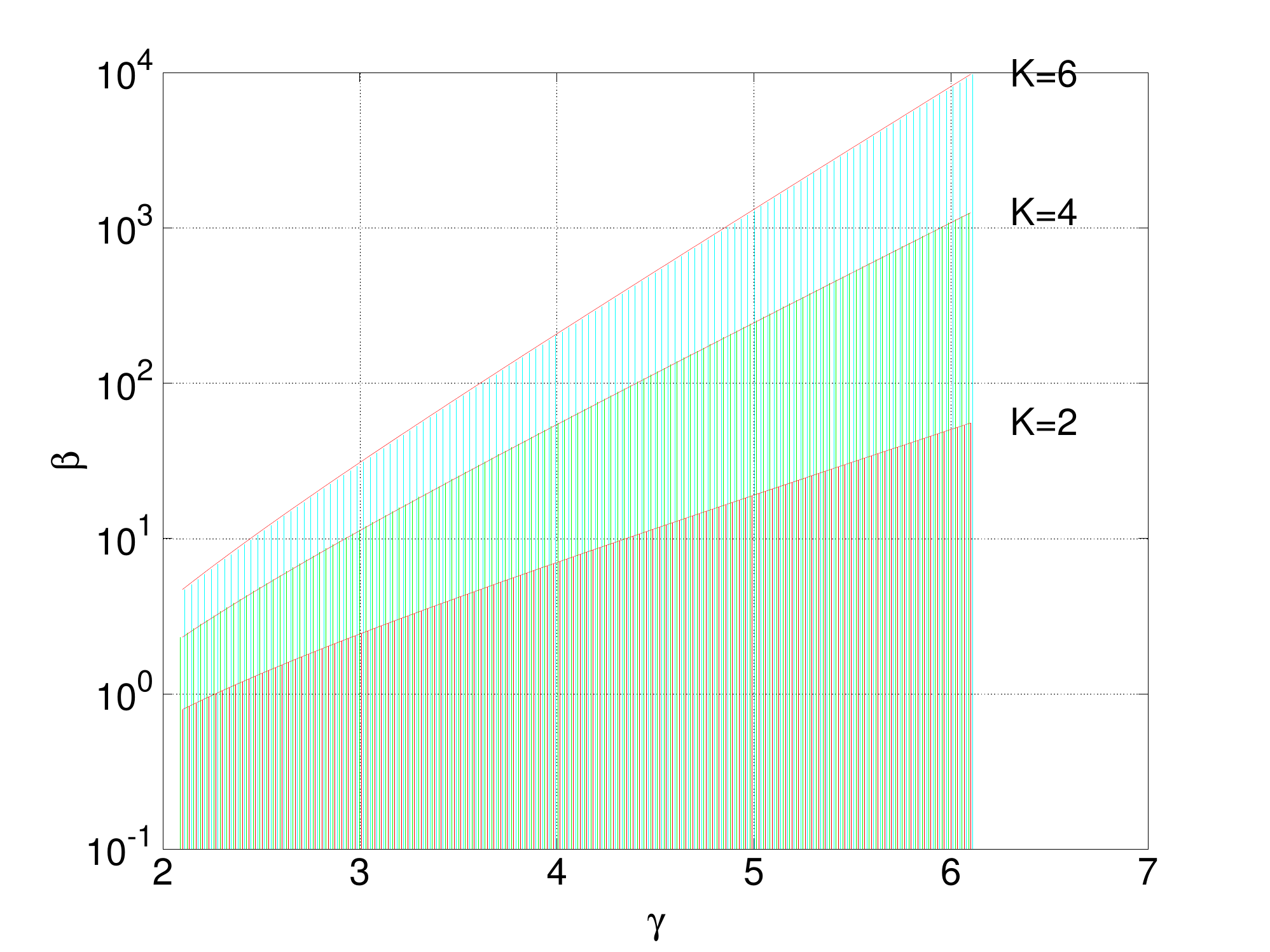}
}
\caption{Feasibility region for hexagonal networks.}
\label{fig:feasr} 
\end{figure}

\begin{figure}[!h]
\centering 
\subfigure[Inequality (\ref{eq:sqdd}), $\beta = 1$] 
{
\label{fig:sqDd}
\includegraphics[width=2.25in]{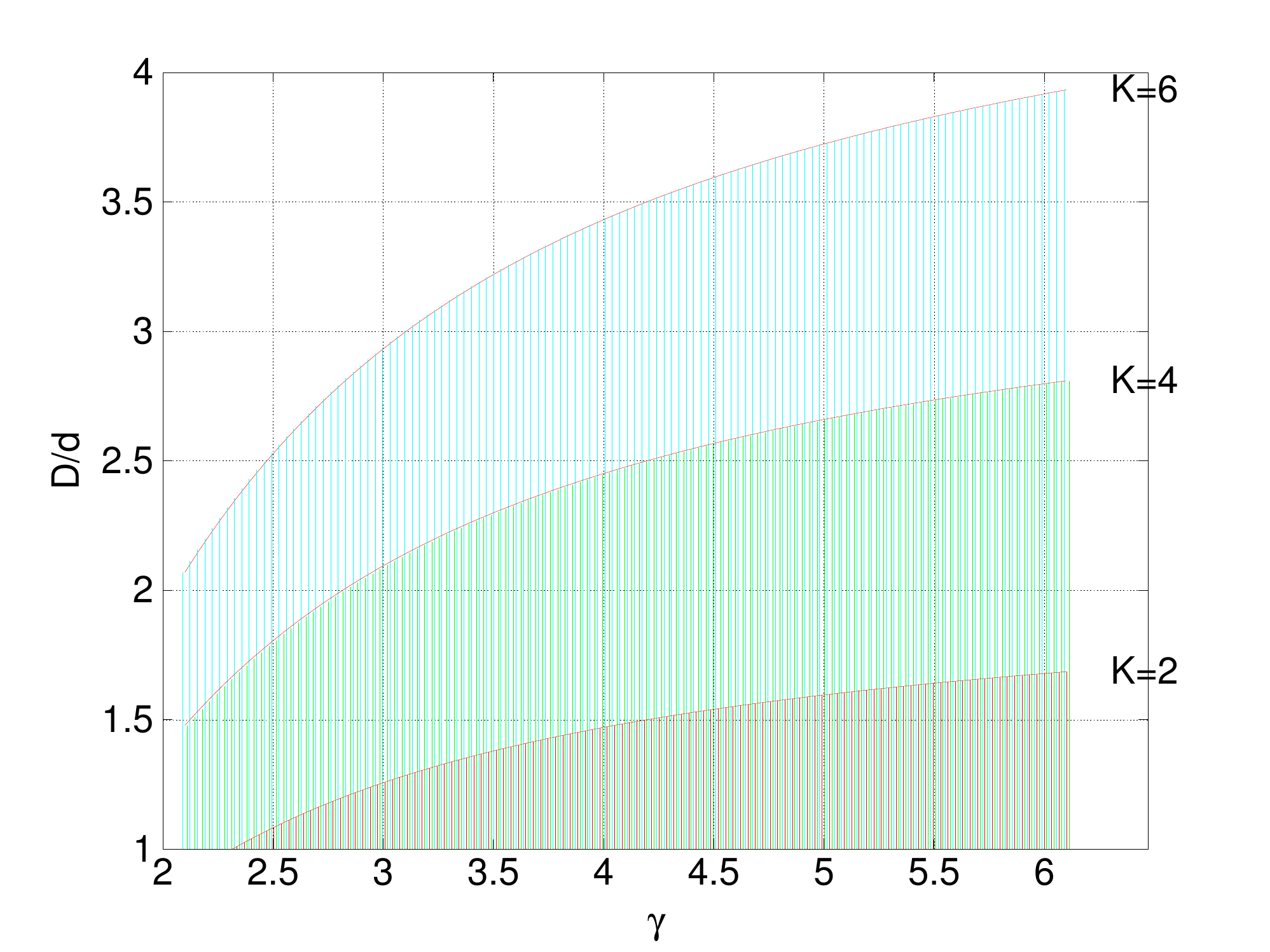}
}
\subfigure[Inequality (\ref{eq:sqbta})] 
{
  \label{fig:sqbeta}
  \includegraphics[width=2.25in]{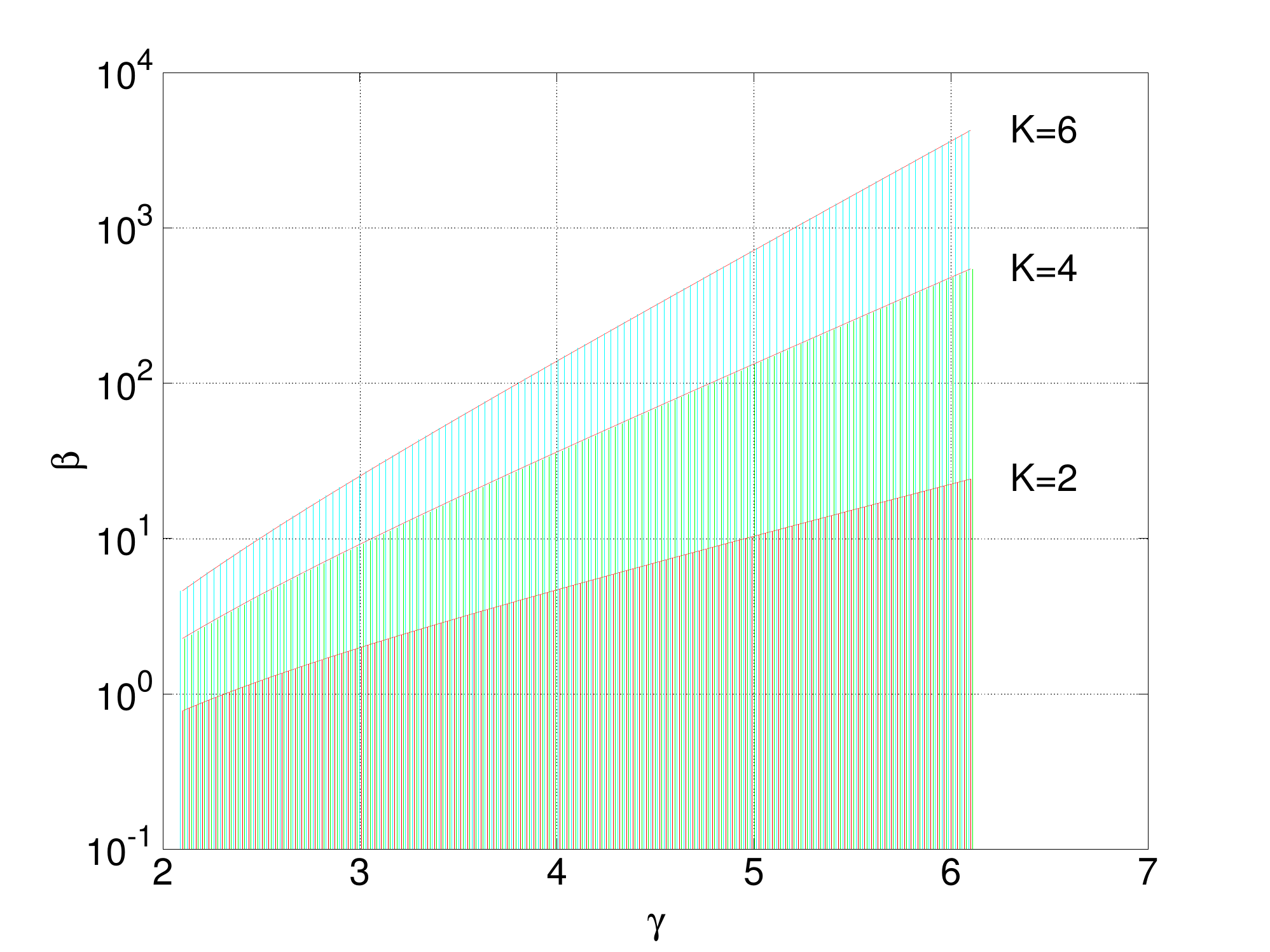}
}
\caption{Feasibility region for square-grid networks. }
\label{fig:feassq} 
\end{figure}

\subsection{Square-Grid Networks} 
We proceed similar to the previous analysis, and use the same symbols as in the hexagonal case. However, we put a superscript $S$ when ambiguity may arise, e.g., $P_S$ denotes the transmission power. For square-grid networks, the set of simultaneously scheduled nodes fall on concentric squares (Figure~\ref{fig:sqconc}). The corners of these concentric squares are at lattice coordinates $(\pm n (k+1), 0)$ and $(0, \pm n(k+1))$. Observe that there are $2n$ transmitters on each of the four sides.
\begin{eqnarray}
s_{i,n}^2  &=& n^2 l^2 + (i-1)^2 (l/\sqrt{2})^2 - 2 (i-1) n l (l/\sqrt{2}) \cos(\pi/4)\nonumber\\
  {\rm or,~} s_{i,n}   &=& l (n^2 + (i-1)^2/2 - (i-1) n)^{1/2}.\nonumber
\end{eqnarray}

Upper bound on interference is:
\begin{eqnarray}
{\rm Interf.}  &\le& 4 \sum_{n=1}^{\infty} \sum_{i=1}^{2n} \frac{P_S}{{(s_{i,n}-d)}^{\gamma}}.\label{eq:interf:sq:h1}
\end{eqnarray}
Unlike hexagonal networks, there is large difference in the distances between the origin and the lattice points at the four corners, and between the origin and the mid-point of a side of the squares. To obtain a tighter bound, we bound the inverse distances of $n/2$ nodes centered at the midpoint by that of the origin and $n^{\rm th}$ lattice point, and bound the inverse distances of the rest $n/2$ nodes of the side by that of the origin and $n/2^{\rm th}$ lattice point (for $n=1$, the latter is the corner point).

\begin{eqnarray}
  s_{n,n} - d  &=& nl\left( \frac{1}{\sqrt{2}} - \frac{1}{\sqrt{2}} \frac{d}{l}  - \frac{d}{nl} + \frac{1}{\sqrt{2}} \frac{d^2}{l^2} + O\left(\frac{d^3}{l^3}\right) \right).\nonumber\\
 s_{n/2,n} - d  &=& nl \left( \frac{\sqrt{5}}{\sqrt{8}} - \frac{3}{\sqrt{40}} \frac{d}{l} - \frac{d}{nl} + \frac{1}{\sqrt{40}} \frac{d^2}{l^2} + O\left(\frac{d^3}{l^3}\right) \right).\nonumber
\end{eqnarray}
Upon simplification:
\begin{eqnarray}
{\rm Interf.}  & < & \frac{4 P_S}{{((k+1)d)}^{\gamma}} \left(\nu^{\gamma} + \phi^{\gamma} \right) \frac{\gamma -1}{\gamma -2},\label{eq:intsq}
\end{eqnarray}
 
where $1/\nu = \frac{1}{\sqrt{2}} - \frac{1}{\sqrt{2}}\left(\frac{1}{k+1}\right)$, and $1/\phi = \frac{\sqrt{5}}{\sqrt{8}} - \frac{3}{\sqrt{40}}\left( \frac{1}{k+1}\right)$. Let us denote $\nu^{\gamma} + \phi^{\gamma}$ by $\alpha$.\\

Thus, from (\ref{eq:sigmin}) and (\ref{eq:intsq}), SINR is above threshold if:

\begin{equation}
P_S > \frac{\beta \eta D^{\gamma}} {1- 4 \alpha \beta {\left(\frac{D}{(k+1)d}\right)}^{\gamma}  \frac{\gamma-1}{\gamma-2}}. \label{eq:sqpossible}
\end{equation}
Since the denominator must be greater than 0, we get:
\begin{equation}
{\left(\frac{D}{d}\right)}_S <  (k+1) {\left(\frac{\gamma-2}{4 \alpha \beta (\gamma-1)}\right)}^{1/\gamma}. \label{eq:sqdd}
\end{equation}
From $D \ge d$ we get, 
\begin{equation}
\beta_S \le  (k+1)^\gamma \frac{\gamma-2}{4 \alpha (\gamma-1)}. \label{eq:sqbta}
\end{equation}
Plots of the feasibility regions defined by (\ref{eq:sqdd})-(\ref{eq:sqbta}) are shown in \figurename~\ref{fig:feassq}.

\subsection{Discussions} 
The inequalities  (\ref{eq:hexpossible})--(\ref{eq:hexbta}) and (\ref{eq:sqpossible})--(\ref{eq:sqbta}) specify feasibility regions in the SINR parameter space where constant (in $|\verts|$) schedule complexity holds. As the figures show, the feasibility region admits large enough $\beta$ for typical $\gamma$ values to be useful in practice. Although the SINR parameters $\gamma$ and $\beta$ may not be assumed to be user-controllable, the free parameter $k$, however, is user-controllable and can be varied to reach a feasible operating point. Since $k$ controls spacing of neighboring simultaneously scheduled nodes, increasing $k$ can make a parameter set feasible. Furthermore, due to the step used for deriving (\ref{eq:hexdd}) and (\ref{eq:sqdd}), the transmit power can be too high at operating points close the feasibility region boundaries. As we show in the next section, increasing $k$ lowers the transmit power level. This reduction is rather dramatic for operating points that are very close to feasibility region boundary.

\section{Evaluation}\label{sec:eval}
While schedule length, determined by the free parameter $k$ is fixed, SINR varies at individual receivers due to irregular physical node placement. Let  $\rho$ be defined as 
\begin{equation}
  \rho = {\rm SINR}/\beta.\nonumber
\end{equation}
Then,  the minimum value of $\rho$, denoted by Min($\rho$), where the minimum is taken over the SINR at all receivers being equal to 1 implies power optimal operating point under uniform transmit power assignment policy. Furthermore, it is desirable that the average $\rho$, denoted by Avg($\rho$), also be close to Min($\rho$). In the following, we present Min($\rho$) as well as the ratio of Avg ($\rho$) and Min($\rho$) obtained by simulation. We generated network topologies where $D$ was normalized to 1 and the node placement irregularity was controlled by 
$D/d$, which was fixed for a given run. 
We determined the physical location of node by using a uniformly distributed random number, $u$,  between 0 and $D/d$ -1 and placing the node at a random point on the circumference of a circle of diameter $u/(2+u)$ centered at a lattice point. 

\subsection{Utilization and Transmit Power Trade-off}\label{ssec:pwr}
In this set of experiments, we started at a feasible operating point close to feasibility region boundary, and measured the variations in Min($\rho$) and the transmit power with the free parameter $k$, while keeping the values of all other SINR and topological parameters fixed. We simulated a network of 4000 nodes. Since larger $k$ implies longer schedule length, this experiment essentially studies the utilization versus transmit power and utilization versus optimality trade-offs. 

\begin{figure}[!h]
\begin{minipage}[t]{\textwidth}
\centering 
\subfigure[Hexagonal] 
{
\label{fig:fixedfrac:hex}
\includegraphics[width=\figwidth]{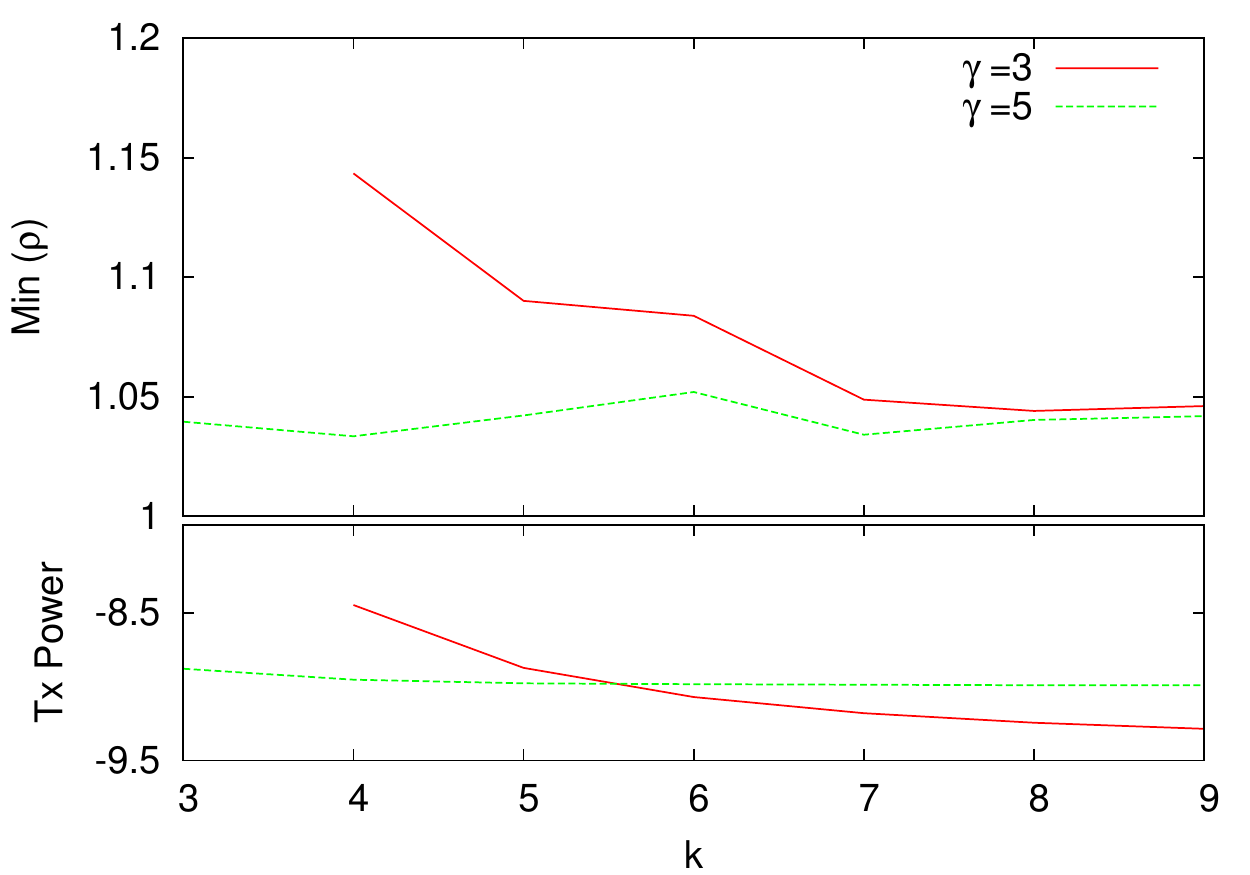}
}
\subfigure[Square-grid] 
{
  \label{fig:fixedfrac:sq}
  \includegraphics[width=\figwidth]{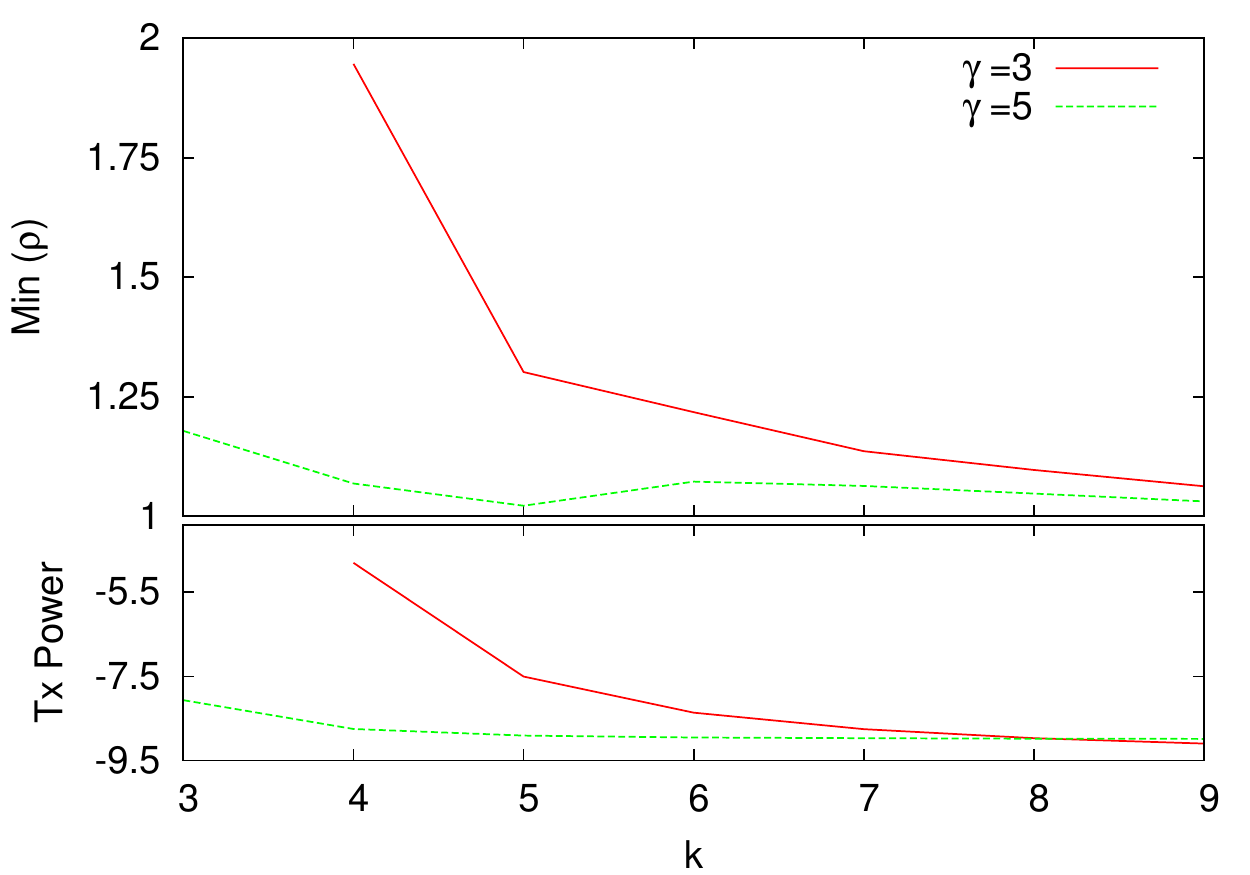}
}
\caption{Utilization vs.\ optimality and power trade-offs.}
\label{fig:fixedfrac:opt} 
\end{minipage}
\end{figure}

We found that Min($\rho$) and transmit power, both, decrease with increasing $k$ (Figs.\ \ref{fig:fixedfrac:hex}--\ref{fig:fixedfrac:sq}), and the decrease gets very steep in the vicinity of feasibility region boundary. At smaller values of $k$'s, the contribution of interference (to SINR) is strong which results in higher transmit power for smaller path-loss exponent, $\gamma$. However, as $k$ gets larger, the contribution of interference decays faster in comparison with signal strength attenuation, which shows-up as the crossing of the transmit power curves.

\subsection{Characterization of Min($\rho$) Variations}
 We used a parameter $0 \le f \le 1$ to select an operating point. The ratio $D/d$ was set to $ 1 + f* ((D/d)_{\rm max} - 1)$ and the SINR threshold was set to $f* \beta_{\rm max}$ as we varied $f$. Recall that the inequalities of the previous section determine $(D/d)_{\rm max}$ and $\beta_{\rm max}$. We repeated these experiments for several values of $\gamma$ keeping $k$ fixed to 2. The Min($\rho$) variations are shown in Fig.~\ref{fig:feasrgn:hex} for hexagonal networks and in Fig.~\ref{fig:feasrgn:sq} for square-grid networks. We then fixed $f$ to 0.5 and varied $k$. The data from these experiments are presented in Fig.~\ref{fig:k:minratio}.  Observe that in the experiment described in Sec.~\ref{ssec:pwr}, while we varied $k$, we kept the other parameters fixed, whereas in this experiment, we vary the other SINR parameters along with $k$ in order to keep the parameters $D/d$ and $\beta$ at their mid-points ($f = 0.5$).

\begin{figure}[!h]
\begin{minipage}[t]{\textwidth}
\centering 
\subfigure[Hexagonal] 
{
\label{fig:feasrgn:hex}
\includegraphics[width=\figwidth]{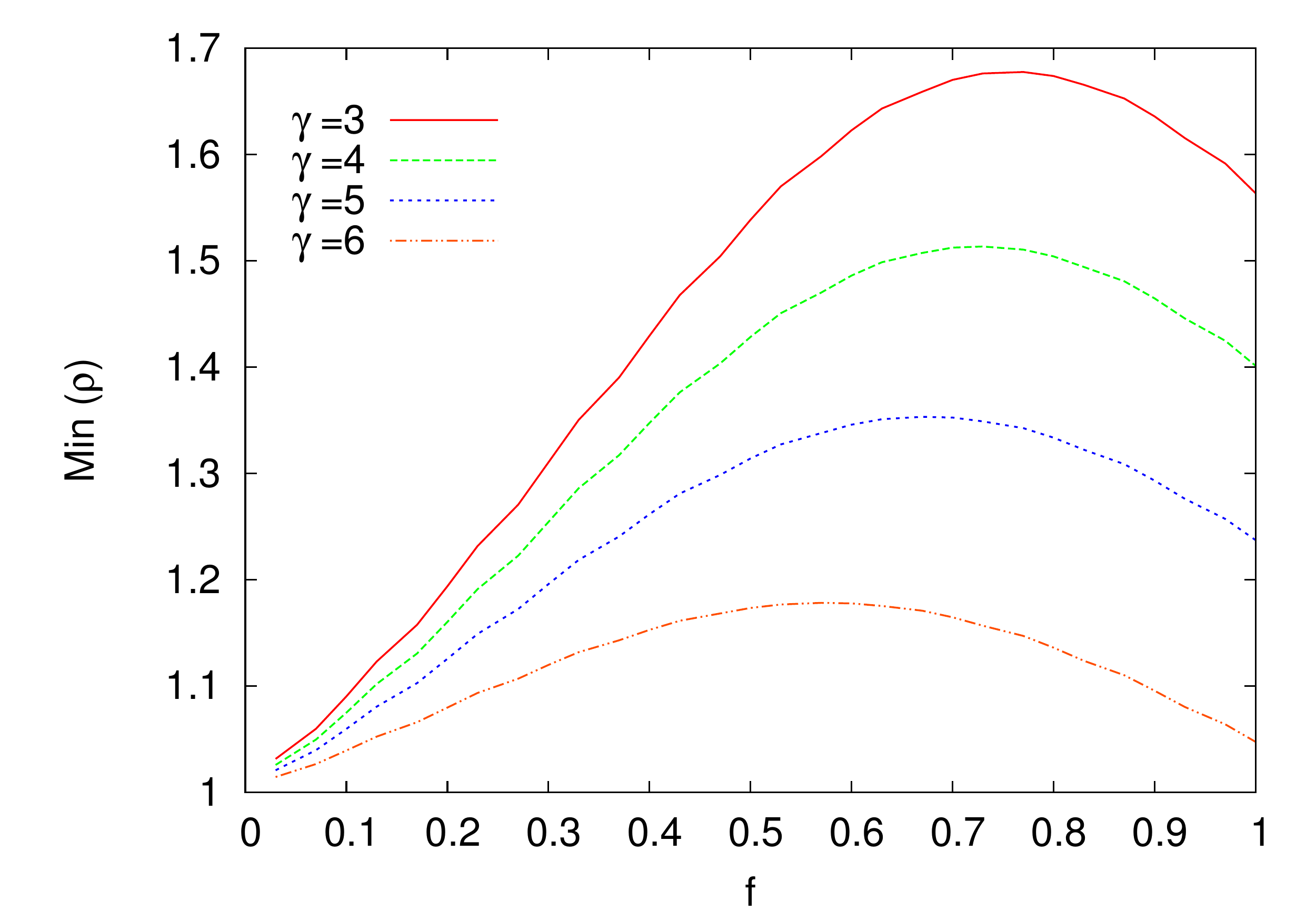}
}
\subfigure[Square-grid] 
{
  \label{fig:feasrgn:sq}
  \includegraphics[width=\figwidth]{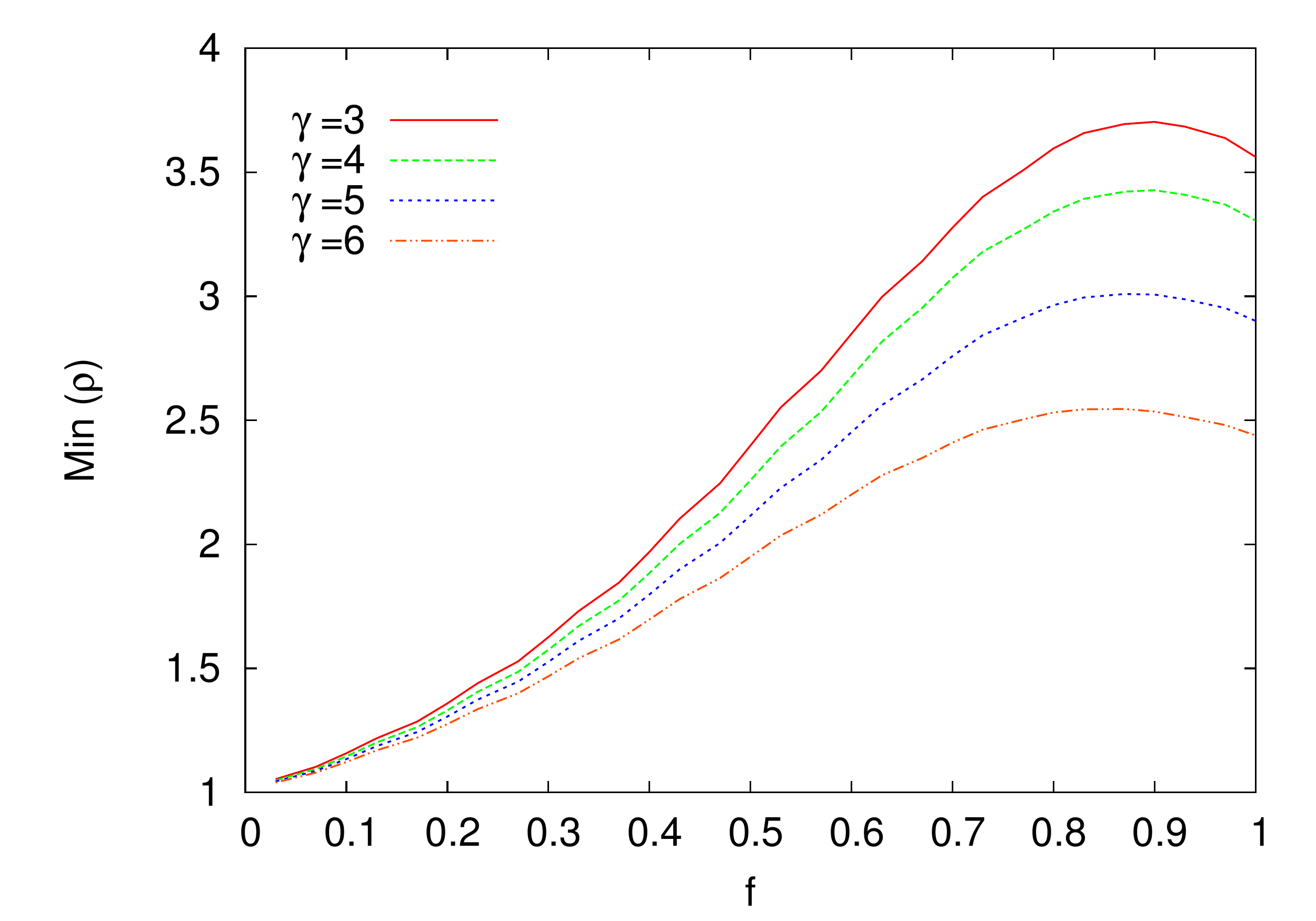}
}
\caption{Min ($\rho_i$) at $k = 2$.}
\label{fig:feasrgn:frac} 
\end{minipage}
\end{figure}

\begin{figure}[!h]
\begin{minipage}[t]{\textwidth}
\centering 
\subfigure[Hexagonal] 
{
\label{fig:k:minratio:hex}
\includegraphics[width=\figwidth]{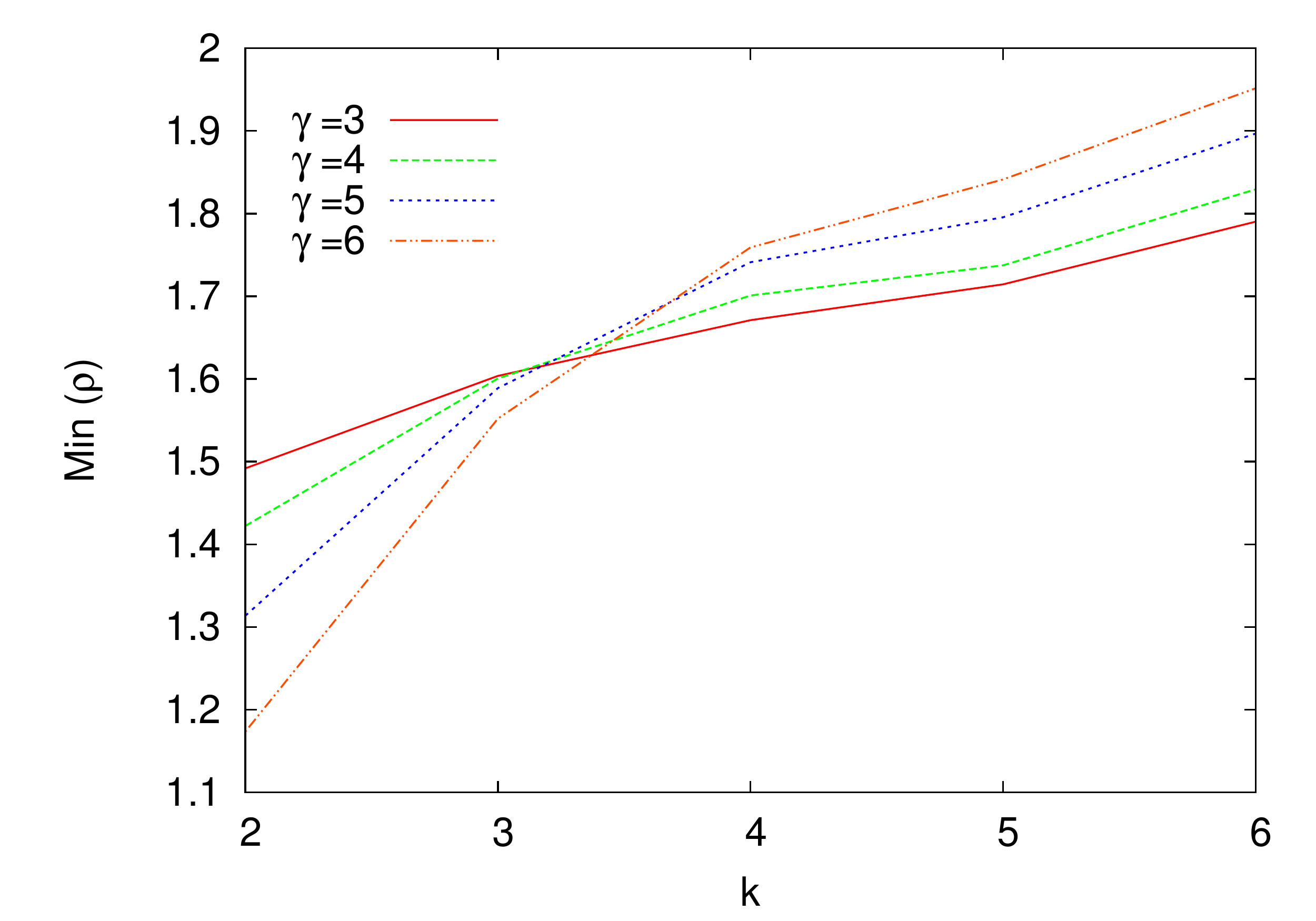}
}
\subfigure[Square-grid] 
{
  \label{fig:k:minratio:sq}
  \includegraphics[width=\figwidth]{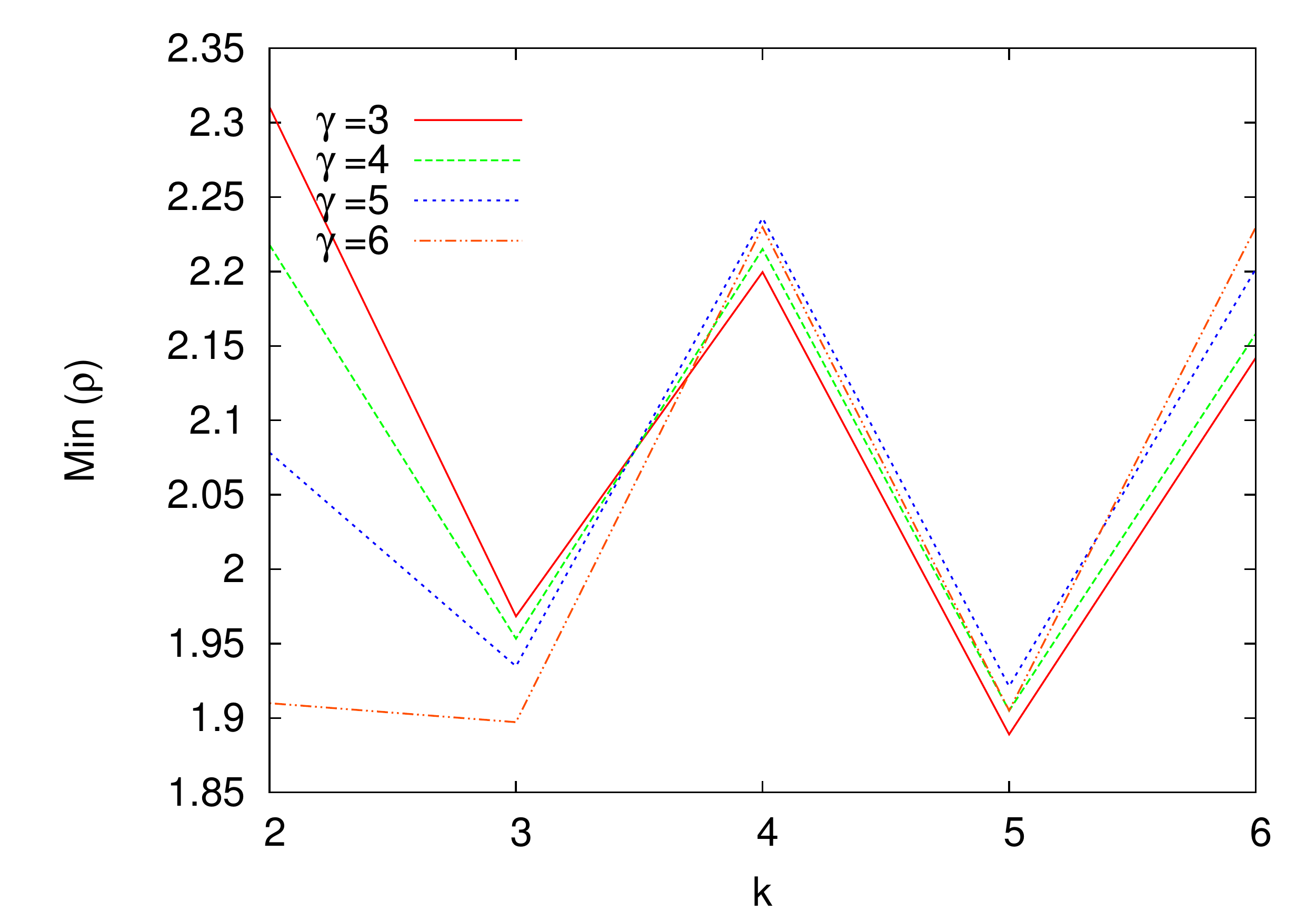}
}
\caption{Min ($\rho_i$) vs.\ $k$, $f=0.5$. }
\label{fig:k:minratio} 
\end{minipage}
\end{figure}

Besides correctness (SINR no less than the threshold: Min($\rho$)~$\ge 1$), the figures show that Min($\rho$) remains within 2 for hexagonal networks and within 4 for square-grid networks. In all cases, as we shift the operating point towards lower $f$, the performance gets better. The deviations from $\rho = 1$ are mainly due to the process of bounding the interference (\ref{eq:intsq}) and due to irregularity of node placements. As $f \rightarrow 0$, $\rho$ approaches 1, and in this part of the feasibility region, the extent of irregularity starts to vanish. The values of Min($\rho$) for square-grid networks are higher than the corresponding ones for hexagonal networks due to larger variations in the transmitter distances for square-grid networks which resulted in less tight bound~(\ref{eq:intsq}). It should be kept in mind that this performance is achieved without message passing overhead for either scheduling or power control. 


Figs.~\ref{fig:feasrgn:minav:frac}-\ref{fig:k:ratio} show the ratio Avg($\rho$) / Min($\rho$). For hexagonal networks, this ratio remains within 1.5 indicating that uniform power assignment is acceptable. The ratio remains similar for $f < 0.5$ in square-grid networks. However, as we approach the feasibility region boundary, these evaluations suggest the need for non-uniform power assignment.
\begin{figure}[!h]
\begin{minipage}[t]{\textwidth}
\centering 
\subfigure[Hexagonal] 
{
\label{fig:feasrgn:minav:hex}
\includegraphics[width=\figwidth]{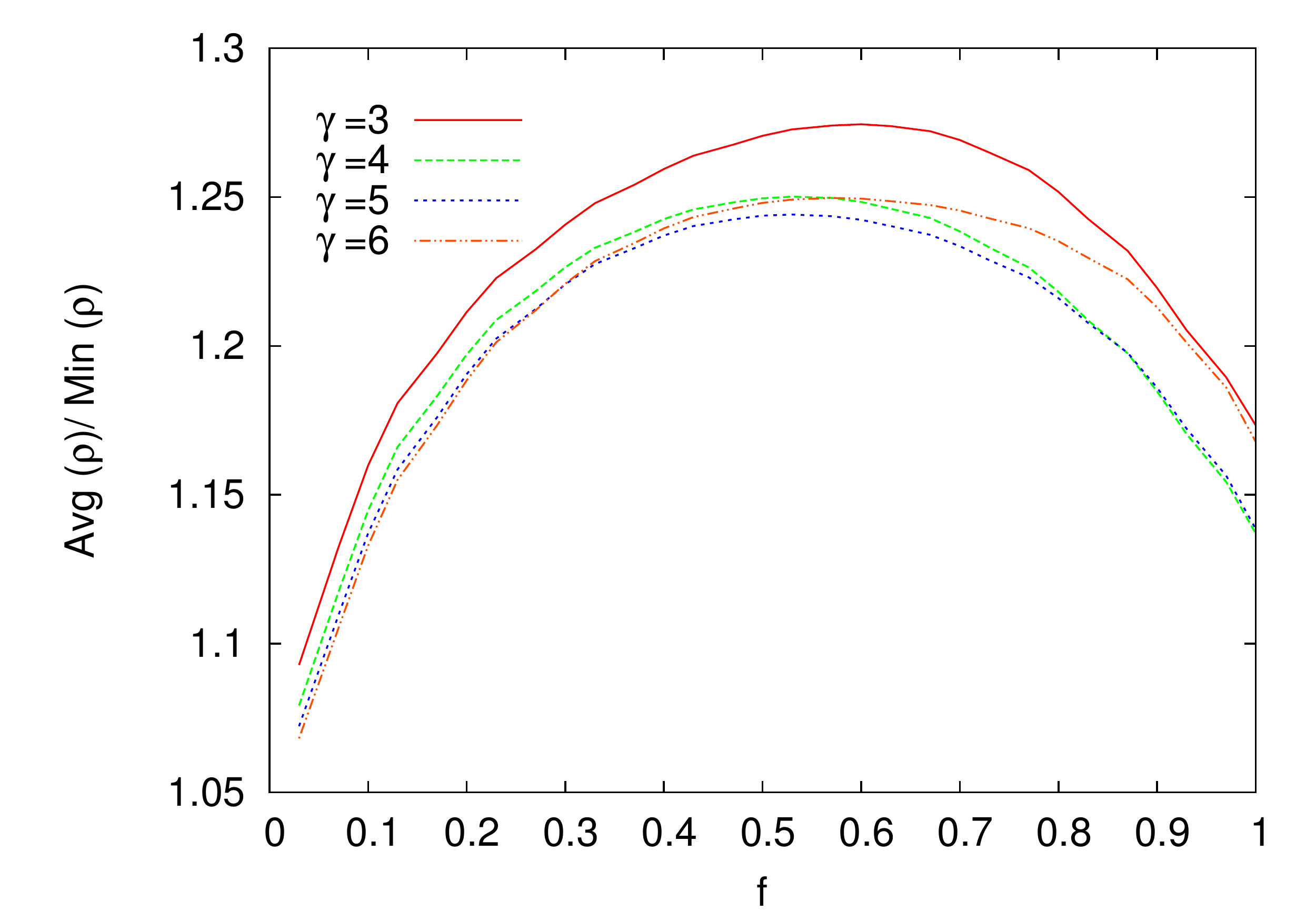}
}
\subfigure[Square-grid] 
{
  \label{fig:feasrgn:minav:sq}
  \includegraphics[width=\figwidth]{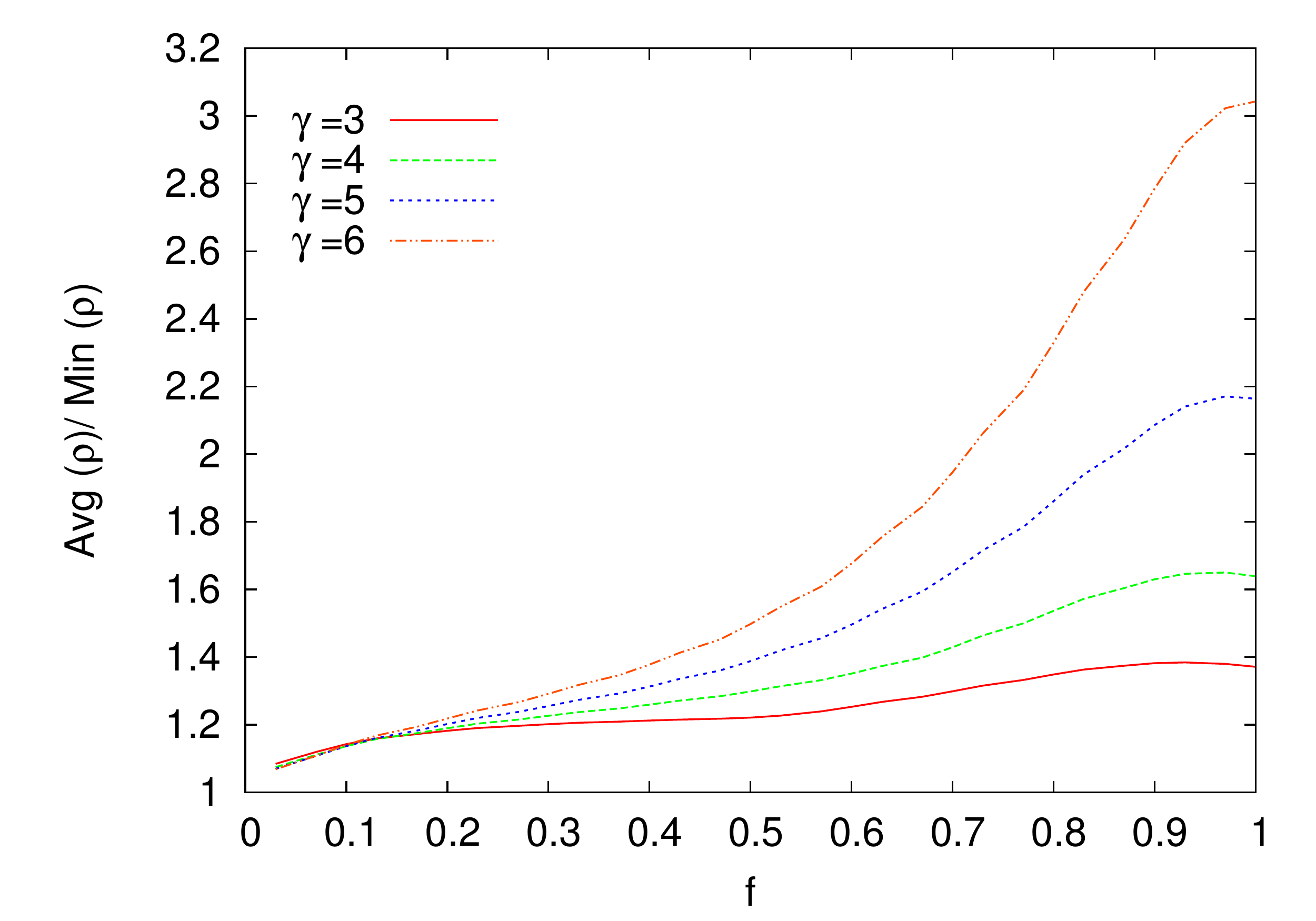}
}
\caption{Avg ($\rho_i$) / Min ($\rho_i$) at $k = 2$.}
\label{fig:feasrgn:minav:frac} 
\end{minipage}
\end{figure}

\begin{figure}[!h]
\begin{minipage}[t]{\textwidth}
\centering 
\subfigure[Hexagonal] 
{
\label{fig:k:ratio:hex}
\includegraphics[width=\figwidth]{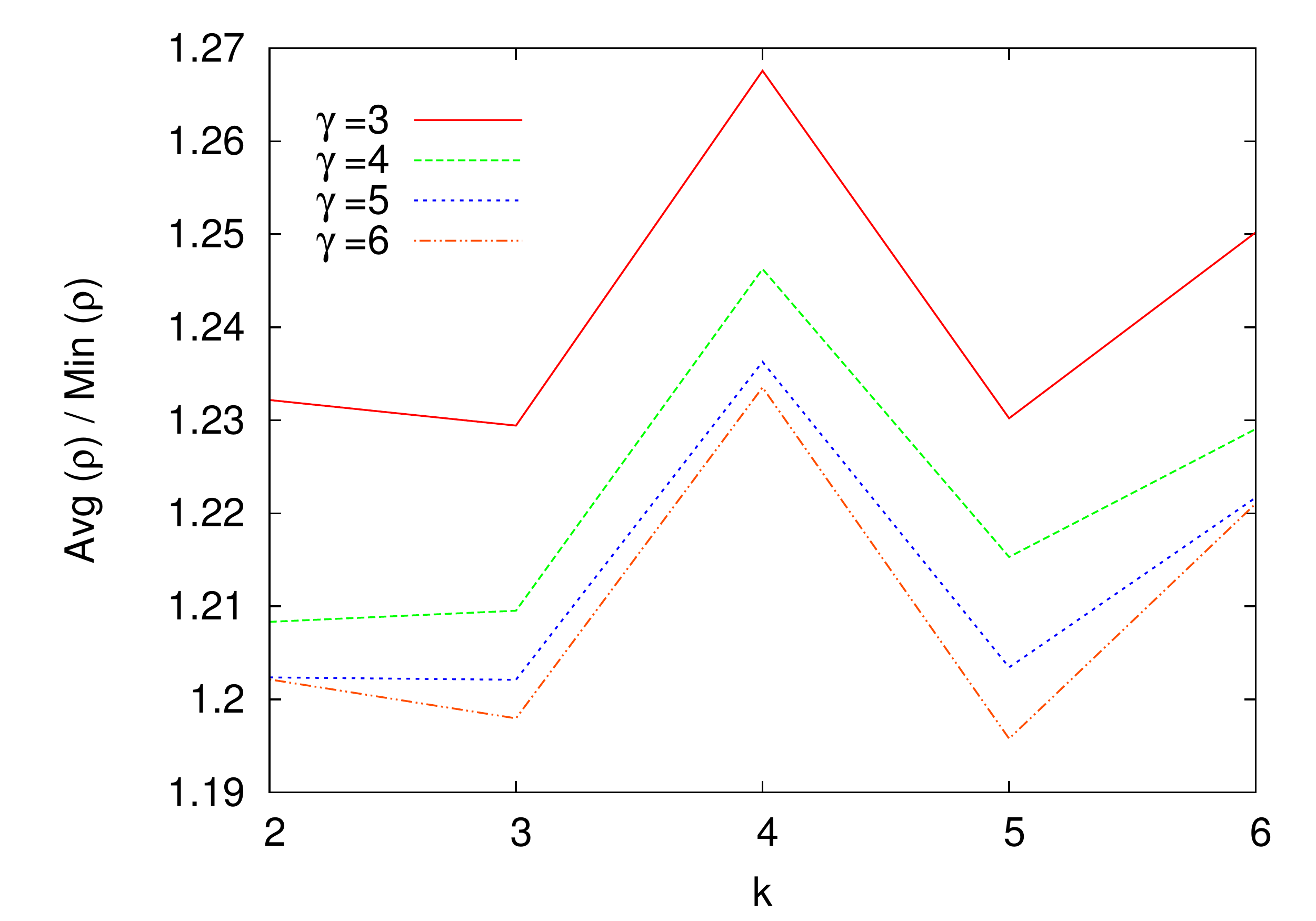}
}
\subfigure[Square-grid] 
{
  \label{fig:k:ratio:sq}
  \includegraphics[width=\figwidth]{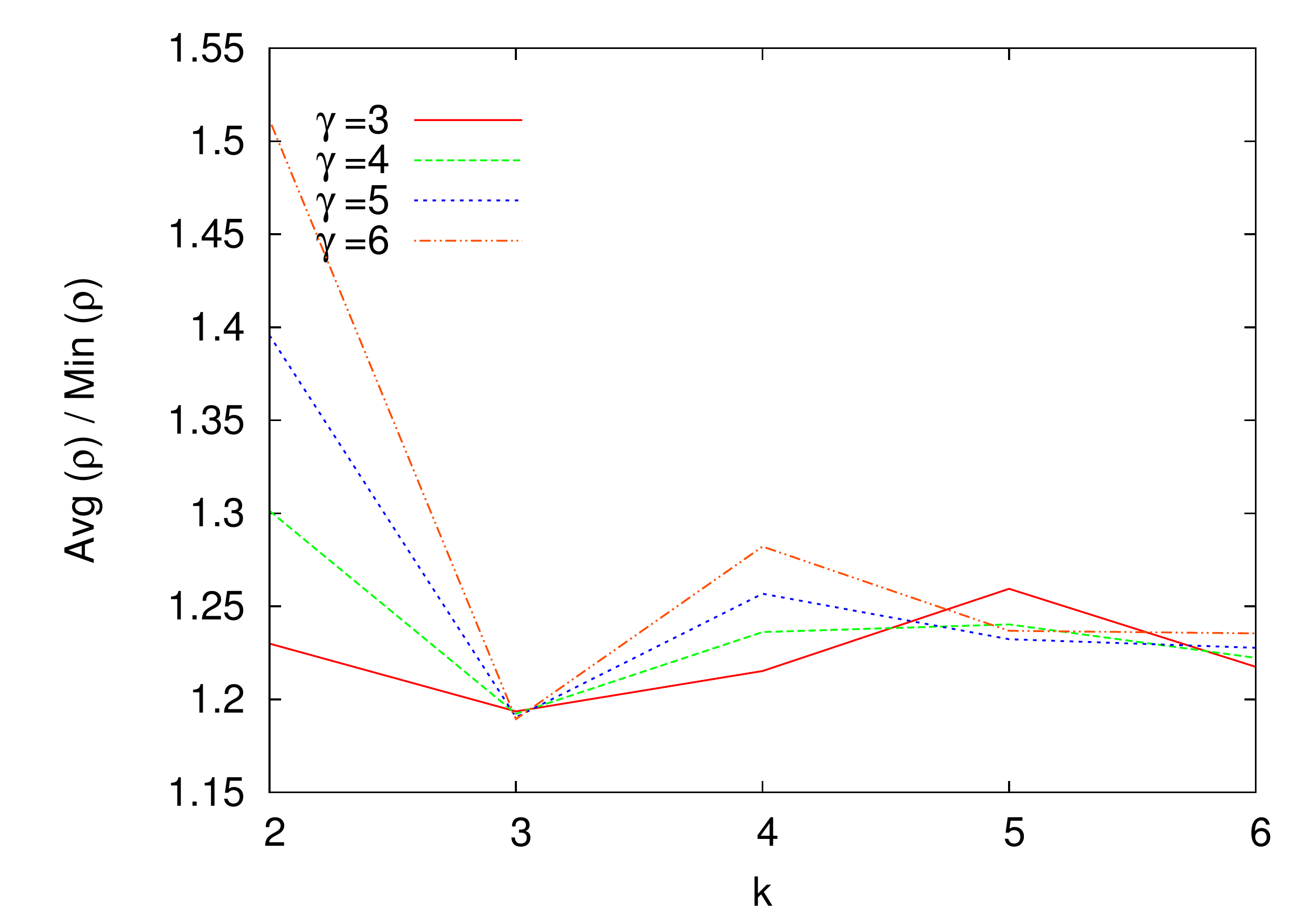}
}
\caption{Avg ($\rho_i$) / Min ($\rho_i$) vs.\ $k$, $f=0.5$. }
\label{fig:k:ratio} 
\end{minipage}
\end{figure}

\subsection{Network Size}
Finally, we measured Min($\rho$) and Avg($\rho$) with increase in network size. As Figs.~\ref{fig:opt:nodecnt:min}--\ref{fig:opt:nodecnt:av} show, both  the minimum and average $\rho$ are non-increasing with increase in scale. The small fluctuations in the second figure, Figs.~\ref{fig:opt:nodecnt:av}, is due to the hexagonal structure being only partially complete. 

\begin{figure}[!h]
\centering 
\subfigure[Min($\rho_i$)] 
{
  \label{fig:opt:nodecnt:min}
  \includegraphics[width=\figwidth]{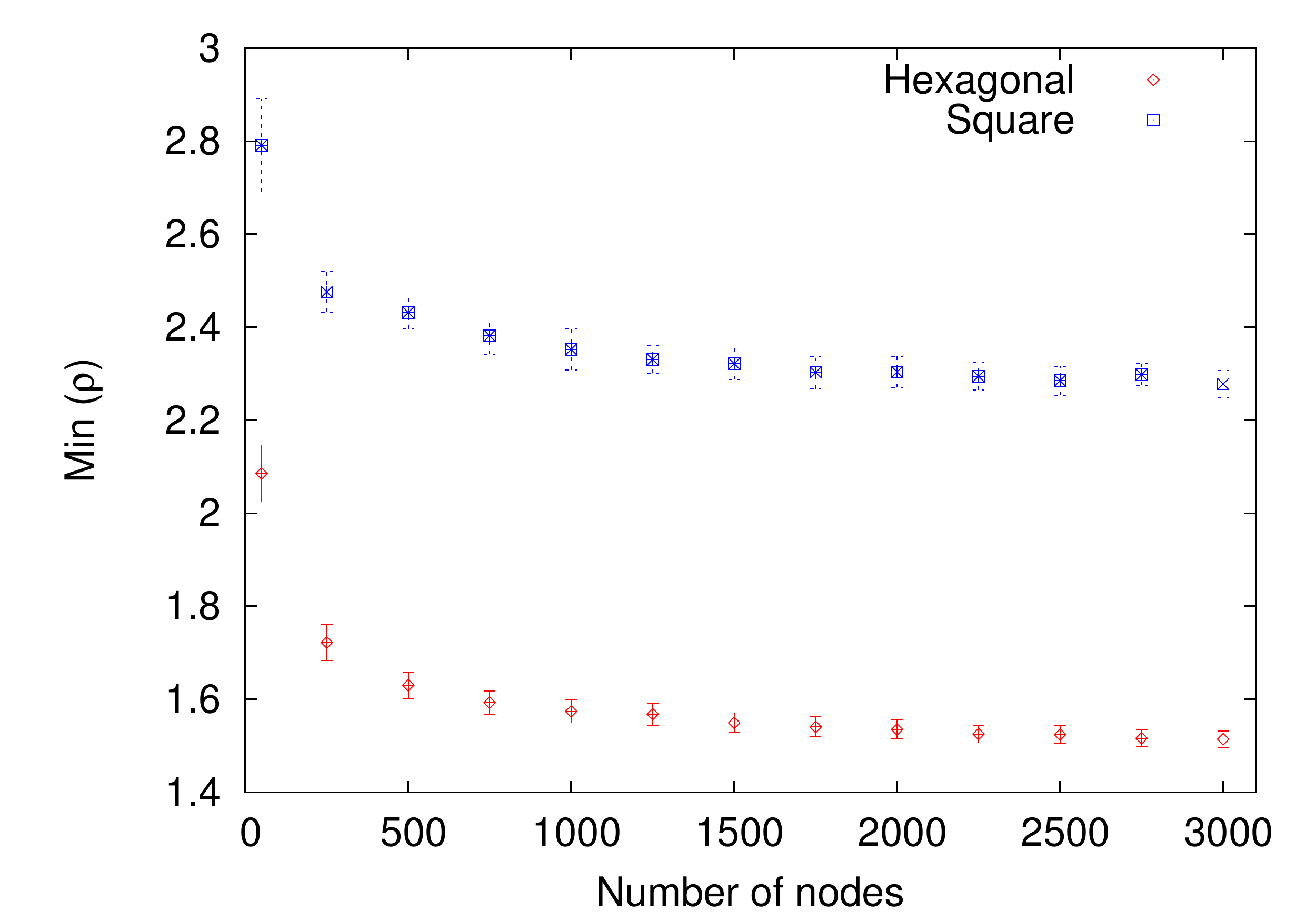}
}
\subfigure[Avg($\rho_i$)] 
{
\label{fig:opt:nodecnt:av}
\includegraphics[width=\figwidth]{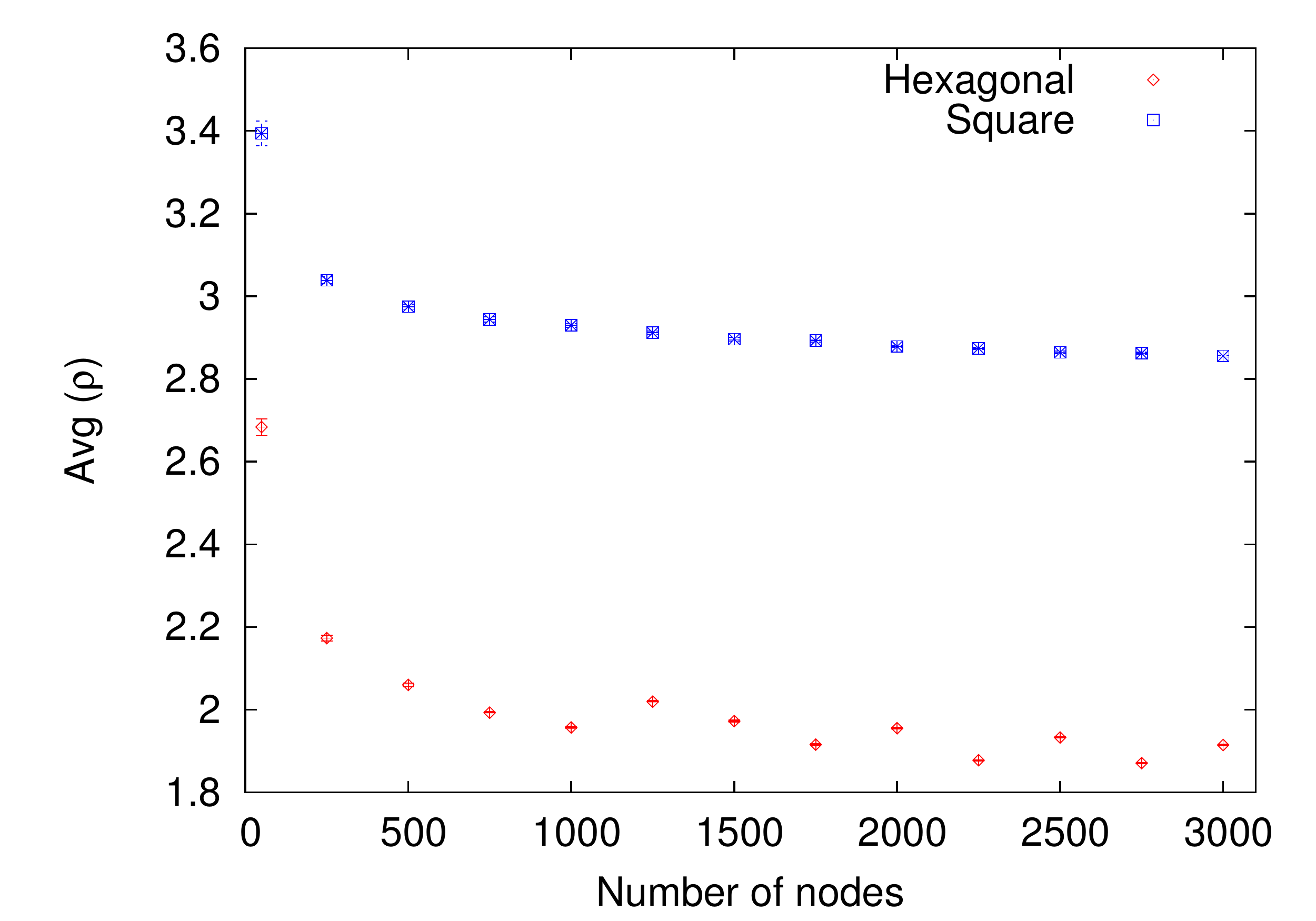}
}
\caption{$\rho$ as a function of node count: $k=2, \gamma = 3, f = 0.5$. }
\label{fig:opt:nodecnt} 
\end{figure}

  \section{Related Work}\label{sec:rel-work}

In one of the earliest works on regular wireless networks, Silvester and Kleinrock investigated the capacity of multi-hop regular topology ALOHA networks~\cite{SK83_multihop_aloha_capacity}. Recently, Mergen and Tong extended this work and analyzed the capacity of regular networks~\cite{MT05_reg_wirel_netw}. In~\cite{MR06_regular_mac}, Mangharam and Rajkumar present a MAC protocol called MAX for square-grid networks with regular node placement.  In all these, transmission and interference ranges are assumed to be the same, an assumption that we have gotten rid of in this paper. In related previous work, the author presented a distributed algorithm for convergecast in hexagonal networks~\cite{PA07_wsn_hex}, and investigated the feasibility of hexagonal backbone formation in sensor network deployments~\cite{Prabh09_hex_prototype}.

McDiarmid and Reed showed that the chromatic number of general hexagonal graphs is bounded by 4/3 times the clique number~\cite{MR00_channel_assignment_4_3_proof}. This result suggests that our hexagonal scheduling algorithm, the complexity of which is also bounded by 4/3 times the clique number, might indeed be very close to optimal. A 4/3-approximate distributed channel assignment algorithm that uses several rounds of message passing is presented in ~\cite{NS01_dist_4_3_graph_coloring_hex}.

The intractability of the minimum length scheduling problem has been established for graph-based models in~\cite{SV82_complexity_max_matching, SMS06_wless_sched_complexity} and for the SINR model in~\cite{GOW07_complexity_sinr}. The scheduling algorithm in~\cite{GOW07_complexity_sinr} as well as several other papers~\cite{GOW07_complexity_sinr,BRS10_sinr_sched,KWES10_sinr_lp_scheduling} are centralized. In~\cite{KV10_probabilistic_dist_sinr_sched, HM11_probabilistic_dist_sinr_sched}, SINR-based distributed scheduling algorithms are presented. These algorithms, however, are not collision-free and the performance guarantee is probabilistic. In~\cite{RJK10_dist_sinr_heuristic}, a heuristic is presented, but no formal analysis is offered. Gronkvist presents slot stealing strategies~\cite{Gronkvist04_sinr_sched}. Power control strategies for cellular networks are presented in~\cite{HBGG03_distr_power_assignment,HRC06_dist_power-control_sinr}.

\section{Conclusion}\label{sec:concl}
In large-scale wireless sensor networks deployments, forming a regular topology backbone is practically feasible. We presented collision-free STDMA-based distributed scheduling algorithms that do not use message passing.  In sensor networks, communication is typically event-triggered and nodes remain in sleep state much of the time to save the limited energy supply. Hence schedulers that use little or no message passing are particularly suited for sensor networks. We showed that the NSP schedule complexity of these algorithms is independent of the number of nodes in the $k$-hop interference model and within certain feasibility regions in SINR model. We  characterized these feasibility regions and saw that typical operating points are well covered by these regions. Investigation of SINR scheduling with non-uniform power assignments and of end-to-end delay are future work of interest.

\appendix
\section{Proof of Lemma 1}
\begin{proof}
 By translating the origin to $(x_1, y_1)$, the desired distance is equal to the distance between the origin and $(x_2-x_1, y_2-y_1)$. The following is a proof by induction on the distance from the origin. 

Base case: The six lattice points at unit distance form the origin are 
\begin{equation}
  (1,0), (1,1), (0,1), (-1,0), (-1,-1),(0, -1)\label{eq:1hop:hex}
\end{equation} 
which satisfy the lemma. 

General case: Let the point $(x,y)$, at distance $n$ from the origin, satisfy $MAX\{|x|,|y|,|x-y|\} = n$.
Let the set $\{(x', y')\}$ be the neighboring lattice points of $(x,y)$ at distance $n+1$ from the origin. Given that $\{(x', y')\}$ is at unit distance from $(x, y)$, $x' \in \{ x, x \pm 1\}$ and $y' \in \{ y, y \pm 1\}$. Since $(x', y')$ is a unit distance farther from origin than $(x,y)$, at-least one of the following holds: \begin{equation}
|x'| = |x| + 1, |y'| = |y| + 1.\label{eq:hind1} 
\end{equation}
However, from (\ref{eq:1hop:hex}), if both $x'$ and $y'$ differ from $x$ and  $y$ respectively, either $(x', y') = (x+1, y+1)$ or $(x', y') = (x-1, y-1)$. Thus, \begin{equation}
|x'-y'| = |x-y| \le n.\label{eq:hind2} 
\end{equation}
 Therefore, from (\ref{eq:hind1}) and (\ref{eq:hind2}), $MAX\{|x'|,|y'|,|x'-y'|\} = n+1$. Hence the lemma follows.
\end{proof}

\bibliographystyle{abbrv}
\bibliography{biblio}

\begin{thebibliography}{10}

\bibitem{BRS10_sinr_sched}
D.~M. Blough, G.~Resta, and P.~Santi.
\newblock Approximation algorithms for wireless link scheduling with
  {SINR}-based interference.
\newblock {\em Networking, IEEE/ACM Transactions on}, 18(6):1701--1712, 2010.

\bibitem{DS05_3d_hex}
C.~Decayeux and D.~Seme.
\newblock {3D} hexagonal network: Modeling, topological properties, addressing
  scheme, and optimal routing algorithm.
\newblock {\em IEEE Trans. Parallel Distrib. Syst.}, 16(9):875--884, 2005.

\bibitem{FDTT07_wless_capacity_percolation}
M.~Franceschetti, O.~Dousse, D.~N.~C. Tse, and P.~Thiran.
\newblock Closing the gap in the capacity of wireless networks via percolation
  theory.
\newblock {\em Information Theory, IEEE Transactions on}, 53(3):1009 --1018,
  march 2007.

\bibitem{GOW07_complexity_sinr}
O.~Goussevskaia, Y.~A. Oswald, and R.~Wattenhofer.
\newblock Complexity in geometric {SINR}.
\newblock In {\em Proceedings of the 8th ACM International Symposium on Mobile
  Ad Hoc Networking and Computing (MobiHoc)}, pages 100--109, 2007.

\bibitem{Gronkvist04_sinr_sched}
J.~Gronkvist.
\newblock Distributed scheduling for mobile ad hoc networks - a novel approach.
\newblock In {\em 15th IEEE International Symposium on Personal, Indoor and
  Mobile Radio Communications (PIMRC)}, volume~2, pages 964 -- 968 Vol.2, sept.
  2004.

\bibitem{HS88_link_sched_poly}
B.~Hajek and G.~Sasaki.
\newblock Link scheduling in polynomial time.
\newblock {\em Information Theory, IEEE Transactions on}, 34(5):910--917, Sep
  1988.

\bibitem{HM11_probabilistic_dist_sinr_sched}
M.~Halld\'{o}rsson and P.~Mitra.
\newblock Nearly optimal bounds for distributed wireless scheduling in the
  {SINR} model.
\newblock In {\em Proceedings of the 38th International Colloquium on Automata,
  Languages and Programming (ICALP)}, volume 6756, pages 625--636. Springer
  Berlin / Heidelberg, 2011.

\bibitem{HRC06_dist_power-control_sinr}
P.~Hande, S.~Rangan, and M.~Chiang.
\newblock Distributed uplink power control for optimal sir assignment in
  cellular data networks.
\newblock In {\em Proceedings of the 25th IEEE International Conference on
  Computer Communications (INFOCOM)}, pages 1 --13, april 2006.

\bibitem{HBGG03_distr_power_assignment}
T.~Holliday, N.~Bambos, P.~W. Glynn, and A.~Goldsmith.
\newblock Distributed power control for time-varying wireless networks:
  Optimality and convergence.
\newblock In {\em Proceedings of the 2003 Allerton Conference on Communication,
  Control and Computing}, pages 1024--1033, 2003.

\bibitem{KLPP11_sinr_wless_topology}
E.~Kantor, Z.~Lotker, M.~Parter, and D.~Peleg.
\newblock The topology of wireless communication.
\newblock In {\em Proceedings of the 43rd annual ACM symposium on Theory of
  computing (STOC)}, pages 383--392, New York, NY, USA, 2011. ACM.

\bibitem{KV10_probabilistic_dist_sinr_sched}
T.~Kesselheim and B.~V\"{o}cking.
\newblock Distributed contention resolution in wireless networks.
\newblock In {\em Proceedings of the 24th international conference on
  Distributed computing (DISC)}, pages 163--178, Berlin, Heidelberg, 2010.
  Springer-Verlag.

\bibitem{KWES10_sinr_lp_scheduling}
S.~Kompella, J.~Wieselthier, A.~Ephremides, H.~Sherali, and G.~Nguyen.
\newblock On optimal {SINR}-based scheduling in multihop wireless networks.
\newblock {\em Networking, IEEE/ACM Transactions on}, 18(6):1713 --1724, Dec.
  2010.

\bibitem{MR06_regular_mac}
R.~Mangharam and R.~Rajkumar.
\newblock {MAX}: A maximal transmission concurrency {MAC} for wireless networks
  with regular structure.
\newblock In {\em Proc.\ of the 3rd International Conference on Broadband
  Communications, Networks and Systems (BROADNETS)}, pages 1--10, Oct. 2006.

\bibitem{MR00_channel_assignment_4_3_proof}
C.~McDiarmid and B.~Reed.
\newblock Channel assignment and weighted coloring.
\newblock {\em Networks}, 36(2):114--117, 2000.

\bibitem{MT05_reg_wirel_netw}
G.~Mergen and L.~Tong.
\newblock Stability and capacity of regular wireless networks.
\newblock {\em Information Theory, IEEE Transactions on}, 51(6):1938--1953,
  June 2005.

\bibitem{MWZ06_complexity_sinr}
T.~Moscibroda, R.~Wattenhofer, and A.~Zollinger.
\newblock Topology control meets {SINR}: the scheduling complexity of arbitrary
  topologies.
\newblock In {\em Proceedings of the 6th ACM International Symposium on Mobile
  Ad Hoc Networking and Computing (MobiHoc)}, pages 310--321, 2006.

\bibitem{NS01_dist_4_3_graph_coloring_hex}
L.~Narayanan and S.~Shende.
\newblock Static frequency assignment in cellular networks.
\newblock {\em Algorithmica}, 29:396--409, 2001.

\bibitem{NK85_STDMA}
R.~Nelson and L.~Kleinrock.
\newblock Spatial {TDMA}: A collision-free multihop channel access protocol.
\newblock {\em Communications, IEEE Transactions on}, 33(9):934 -- 944, Sept.
  1985.

\bibitem{Peleg_sinr_maps}
D.~Peleg.
\newblock {SINR} maps: Properties and applications.
\newblock In {\em Proceedings of the 18th International Colloquium on
  Structural Information and Communication Complexity (SIROCCO)}, pages 15--16,
  2011.

\bibitem{PA07_wsn_hex}
K.~S. Prabh and T.~Abdelzaher.
\newblock On scheduling and real-time capacity of hexagonal wireless sensor
  networks.
\newblock In {\em Proc.\ of the 19th Euromicro Conference on Real-Time Systems
  (ECRTS)}, pages 136--145. IEEE Press, Los Alamitos, CA, 2007.

\bibitem{Prabh09_hex_prototype}
K.~S. Prabh, C.~Deshmukh, and S.~Sachan.
\newblock A distributed algorithm for hexagonal topology formation in wireless
  sensor networks.
\newblock In {\em Proc.\ of the 14th IEEE Intl. Conf. on Emerging Technologies
  and Factory Automation (ETFA)}. IEEE Press, Los Alamitos, CA, 2009.

\bibitem{Rappaprt01_wless_book}
T.~Rappaport.
\newblock {\em Wireless Communications: Principles and Practice}.
\newblock Prentice Hall PTR, Upper Saddle River, NJ, USA, 2001.

\bibitem{RJK10_dist_sinr_heuristic}
J.~Ryu, C.~Joo, T.~T. Kwon, N.~B. Shroff, and Y.~Choi.
\newblock Distributed {SINR} based scheduling algorithm for multi-hop wireless
  networks.
\newblock In {\em Proceedings of the 13th ACM international conference on
  Modeling, analysis, and simulation of wireless and mobile systems (MSWIM)},
  pages 376--380, New York, NY, USA, 2010. ACM.

\bibitem{SMS06_wless_sched_complexity}
G.~Sharma, R.~R. Mazumdar, and N.~B. Shroff.
\newblock On the complexity of scheduling in wireless networks.
\newblock In {\em Proceedings of the 12th annual international conference on
  Mobile computing and networking (MobiCom)}, pages 227--238. IEEE, 2006.

\bibitem{SK83_multihop_aloha_capacity}
J.~Silvester and L.~Kleinrock.
\newblock On the capacity of multihop slotted aloha networks with regular
  structure.
\newblock {\em IEEE Transactions on Communications}, 31(8):974--982, Aug 1983.

\bibitem{SV82_complexity_max_matching}
L.~J. Stockmeyer and V.~V. Vazirani.
\newblock {NP-Completeness} of some generalizations of the maximum matching
  problem.
\newblock {\em Inf. Process. Lett.}, 15(1):14--19, 1982.

\end{thebibliography}

\end{document}